\documentclass[10pt,journal]{IEEEtran}
\usepackage{cite}
  \usepackage{graphicx}
  \usepackage{subfigure}
  \usepackage{multirow}
  \usepackage{diagbox}
  \usepackage{booktabs}
  \usepackage{amsthm}
    \usepackage{mathtools}
  \usepackage{algorithm}
  \usepackage{algorithmic}
\usepackage{amsmath}
\usepackage{algorithmic}
\usepackage{array}
\usepackage{fixltx2e}

\usepackage{authblk}
\usepackage{amsfonts}
\newcommand{\figref}[1]{\figurename~\ref{#1}}
\usepackage[utf8]{inputenc}
\DeclareUnicodeCharacter{FB01}{fi}
\DeclareUnicodeCharacter{FB03}{ffi}
\newtheorem{remark}{Remark}
\newtheorem{lemma}{Lemma}
\newtheorem{theorem}{Theorem}
\newtheorem{proposition}{Proposition}
\begin{document}
%\title{An ADMM Algorithm for Joint Rate Allocation and Path Selection in WANs
%based on  Multi-Objective Optimization with Path Cardinality Constraints
%}
\title{ Joint Bandwidth Allocation and Path Selection in WANs
 with Path Cardinality Constraints
}

%% 写明作者的名字和机构 城市 国家 邮件地址
%	\IEEEauthorblockA{\textit{dept. name of organization (of Aff.)} \\
%		\textit{name of organization (of Aff.)}\\
%		City, Country \\
%		email address}
%	
%	\and   %连接多个作者
%	\IEEEauthorblockN{2\textsuperscript{nd} Given Name Surname}
%	\IEEEauthorblockA{\textit{dept. name of organization (of Aff.)} \\
%		\textit{name of organization (of Aff.)}\\
%		City, Country \\
%		email address}
%	}

\author[1,2]{Jinxin~Wang}
\author[2]{Fan~Zhang}
\author[3]{Zhonglin~Xie}
\author[2]{Gong~Zhang}
\author[4]{Zaiwen~Wen}
\affil[1]{Center for Data Science, Peking University, Beijing, China}
\affil[2]{Theoretical Research Lab, Huawei Technologies, Hong Kong SAR, China}
\affil[3]{School of Mathematical Sciences, Peking University, Beijing, China}
\affil[4]{Beijing International Center for Mathematical Research, Peking University, Beijing, China}
\renewcommand\Authands{, and }

\markboth{%IEEE TRANSACTIONS ON SIGNAL PROCESSING,~Vol.~40, No.~8, January~15~2020
}%
{Shell \MakeLowercase{\textit{et al.}}: Bare Demo of IEEEtran.cls for IEEE Journals}



\maketitle

\begin{abstract}
In this paper, we study a joint bandwidth allocation and path selection problem via solving a multi-objective minimization problem under the path cardinality constraints, namely MOPC. Our problem formulation captures various types of objectives including the proportional fairness, the total completion time, as well as the worst-case link utilization ratio. Such an optimization problem is very challenging since it is highly non-convex. Almost all existing works deal with such a problem using relaxation techniques to transform it to be a convex optimization problem. However, we provide a novel solution framework based on the
classic alternating direction method of multipliers (ADMM) approach for solving this problem. Our proposed algorithm is simple and easy to be implemented. Each step of our algorithm consists of either finding
the maximal root of a single-cubic equation which is guaranteed to have at least one positive solution or solving a one-dimensional convex subproblem in a fixed interval. Under some mild assumptions, we prove that any limiting point of the generated sequence under our proposed algorithm is a stationary point. Extensive numerical simulations are performed to demonstrate the advantages of our algorithm compared with various baselines. 
\end{abstract}


% \fi
\IEEEpeerreviewmaketitle

\section{Introduction}
\IEEEPARstart{I}{n} wide area networks (WANs), network resource are shared by applications with different high-level performance objectives. For example, high-definition video applications require high throughput for the end users to enjoy video smoothly \cite{kimball2016low}. Online gaming applications require low latency for better user game experience \cite{lee2015outatime}. AR/VR-based applications such as video conferencing require both high throughput and low latency for data transmission \cite{elbamby2018toward}. Note that it is quite challenging to fulfill various performance objectives for these applications. Firstly, there are very limited available bandwidth resources in the network and it is hard to strike a balance across flows with different performance objectives. Secondly, some applications have the advantage of using more than one end-to-end path to deliver the traffic data. This design space of the number of paths also complicates the problem on how to distribute the allocated bandwidth across multiple paths for such flows. Thirdly, there are some underlying limitations on the number of paths an application could use. Such cardinality nature of path selection makes the problem even challenging. In this paper, we will try to address the aforementioned challenges by studying a multi-objective bandwidth allocation and path selection problem with path cardinality constraints.

There are many existing literatures discussing bandwidth allocation and path selection problems in WANs. In \cite{tian2017multi}, an approach called stacked congestion control is proposed to optimize multi-tenant multi-objective performance with a distributed host-based implementation. In \cite{wang2018joint}, a traffic engineering approach in software-defined WANs is proposed to proactively enforce forwarding policies by
coordinating the traffic demand of each data center instead of traditional passive control policies relying on given priorities. In \cite{perry2014fastpass}, a centralized
arbiter is used to decide when and at which path each packet should be sent. Even though these approaches demonstrate performance gains compared with baselines therein, they are purely heuristic-based and lack theoretical performance guarantee. A well-known theoretical framework for solving network resource allocation problems is network utility maximization (NUM). In \cite{nagaraj2016numfabric}, a NUM-based single-objective network problem is solved theoretically with a system-level implementation. However, the approach in \cite{nagaraj2016numfabric} can only solve single-objective optimization problems and fails to address the multi-objective requirements in our problem. Furthermore, the above literatures \cite{tian2017multi,wang2018joint,perry2014fastpass,nagaraj2016numfabric} do not address the issue of path cardinality constraints and hence, can not be extended to solve our problem.

 The formulation proposed in this paper is built on the foundation of the NUM. Note that without path cardinality constraints, the multi-objective multi-path network resource allocation problem is in fact a convex optimization problem, which can be easily tackled using existing optimization softwares such as \textit{cvx} \cite{grant2014cvx}. However, if we add the path cardinality constraints to the NUM framework, such a problem becomes an $\ell_0$-norm constrained optimization problem, which makes the optimization problem very difficult to solve. The practical meaning of the $\ell_0$-norm constraints is to limit the number of paths an application could use. For example, for data backup service, the end users are not sensitive to the arrival sequence of the data packets and the senders can push as much data as possible to the end users with no limitation on the number of paths. For video streaming applications, the user experience heavily depends on the arrival sequence of the data packets and such applications are limited to use just one path to avoid out-of-sequence issue at the end users.

There are some recent works that take path cardinality constraints into consideration in their network resource allocation problems. In \cite{bi2016network}, a NUM problem with linear objective function and path cardinality constraints is considered. The path cardinality constraints are relaxed using a linear envelope and the solution to the relaxed problem is mapped to a vertex solution using a randomization approach. A theoretical performance gap is analyzed compared with the optimal solution to the relaxed model. However, the approximation and analysis therein cannot be extended to the scenarios with nonlinear objective functions. In \cite{xu2019multiple}, a network optimization problem with restrictions in the number of paths and the requirement of delay is considered, which is a mix-integer optimization problem. The authors propose a heuristic solution by first carefully designing the price of each path and develop an algorithm to match the demand of the traffic flows based on ranking the designed prices. However, it is rather hard to have any theoretical analysis on the performance. Note that the above works handle the cardinality constraints by either considering theoretical approximations or heuristic solutions to bypass the associated constraints, which are highly suboptimal and lack performance guarantee. 

 In this paper, we propose a comprehensive framework for the multi-objective bandwidth allocation and path selection problem with path cardinality constraints. To address the above challenges, we use the ADMM approach to decompose the multi-objective problem into several single-objective subproblems, where we explicitly consider the path cardinality constraints. Then we show that these subproblems are in fact simple to solve even though they are non-convex or quadratic programs. We then theoretically prove the convergence of the proposed algorithm under certain mild conditions and show that it at least achieves suboptimality. Finally, by comparing with several baselines, we show the proposed algorithm has a significant performance gain.

The remainder of the paper is organized as follows. The system model is introduced in Section \ref{Sec: 2}. Then two
types of bandwidth allocation problems are formulated in Section \ref{Sec: 3}. In Section \ref{Sec: 4}, we propose an algorithm to solve the two formulated problems. The convergence analysis is given in Section \ref{Sec: 5}. In Section \ref{Sec: 6}, we give the numerical simulations and conclude this paper in Section \ref{Sec: 7}.

\section{SYSTEM MODEL}
\label{Sec: 2}
We consider a general communication network which is used to deliver data flows from source nodes to destination nodes to facilitate end-to-end data services. Specifically, we consider $K$ flows in the network indexed
by $k$ $(k=1,2,\cdots,K)$. 
% \subsection{Network Topology}
There are $L$ uni-directional links indexed by $l$ $(l=1,2,\cdots,L)$ connecting
the network nodes with link capacity ${c}_l$ $({c}_l>0, \forall l)$. We define $\boldsymbol{c}=(c_1, c_2, \cdots, c_L)^\top$ as the link capacity vector of the overall network. The centralized network controller is used to govern the overall bandwidth allocation as shown in \figref{fig: archtect}. Specifically, given the utility of each flow (to be illustrated in Section \ref{Sec: 3}), the network topology and path cardinality constraints, the network controller will allocate bandwidth and paths to each flow so as to maximize the utilities and achieve load balancing. 
Furthermore, there is an associated source-destination pair for each $k$-th flow and we assume it will carry data with size $s_k$ (in bits). Moreover, for each $k$-th flow, there are $P_k$ available paths indexed by $i$ $(i=1,2,\cdots,P_k,\forall k)$. Let $\boldsymbol{x}_k = (x_{k,1},x_{k,2},\cdots,x_{k,P_k})^\top$ be the bandwidth allocation and path selection vector of the $k$-th flow, where each element $x_{k,i}$ $(x_{k,i}\geq 0, \forall k,i )$ represents the bandwidth allocated at path $i$ for the $k$-th flow (in bits/sec). It can be easily observed that the paths selected by the $k$-th flow are those with positive $x_{k,i}$, i.e., $\{i | x_{k,i} > 0  \}$. We further define
\begin{align*}
\boldsymbol{x} = 
\begin{pmatrix}
\boldsymbol{x}_1 \\
\boldsymbol{x}_2 \\
\vdots \\
\boldsymbol{x}_K
\end{pmatrix} 
\end{align*}
as the bandwidth allocation and path selection vector of all the flows. We then define the routing matrix, which is a $L\times P_k$ matrix representing the relationship between the links and the available paths of the $k$-th flow
\begin{equation*}
\boldsymbol{R}_k={
	\left( \begin{array}{cccc}
	R^k_{1,1} & R^k_{1,2} & \cdots & R^k_{1,P_k}\\
	R^k_{2,1} & R^k_{2,2} & \cdots &R^k_{2,P_k}\\
	\vdots & \vdots & \ddots & \vdots \\
	R^k_{L,1} & R^k_{L,2} & \cdots & R^k_{L,P_k}
	\end{array} 
	\right )},
\end{equation*}
where $R^k_{l,i}  \in \{ 0,1\}$. The value $R^k_{l,i}=1$ means path $i$ passes link $l$ for the $k$-th flow and vice versa. We define the following $L \times P$ matrix (where we denote $ P = \Sigma_{k=1}^K P_k $) to be the routing matrix of the overall network:
\begin{align}
\boldsymbol{R} = (\boldsymbol{R}_1, \boldsymbol{R}_2,\cdots, \boldsymbol{R}_K).
\label{eq: routingmatR}
\end{align}
\begin{figure}[]%%图
	\centering  %插入的图片居中表示
	\includegraphics[scale=0.35]{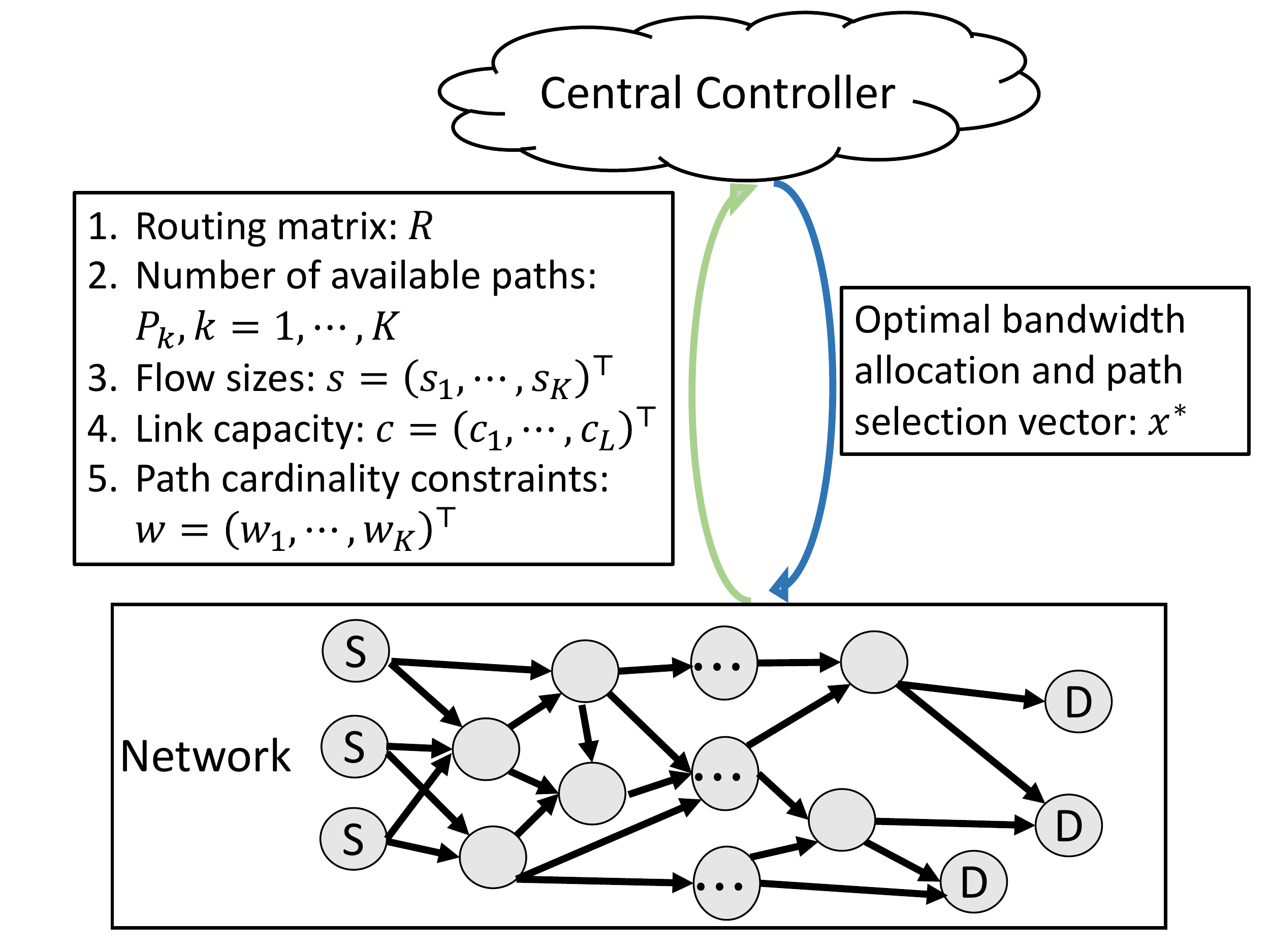} 
	\caption{The architecture of our system model with a centralized controller.}  %图片的名称
	\label{fig: archtect}   %标签，用作引用
\end{figure}
\begin{figure}[h]%%图
	\centering  %插入的图片居中表示
	\includegraphics[scale=0.3]{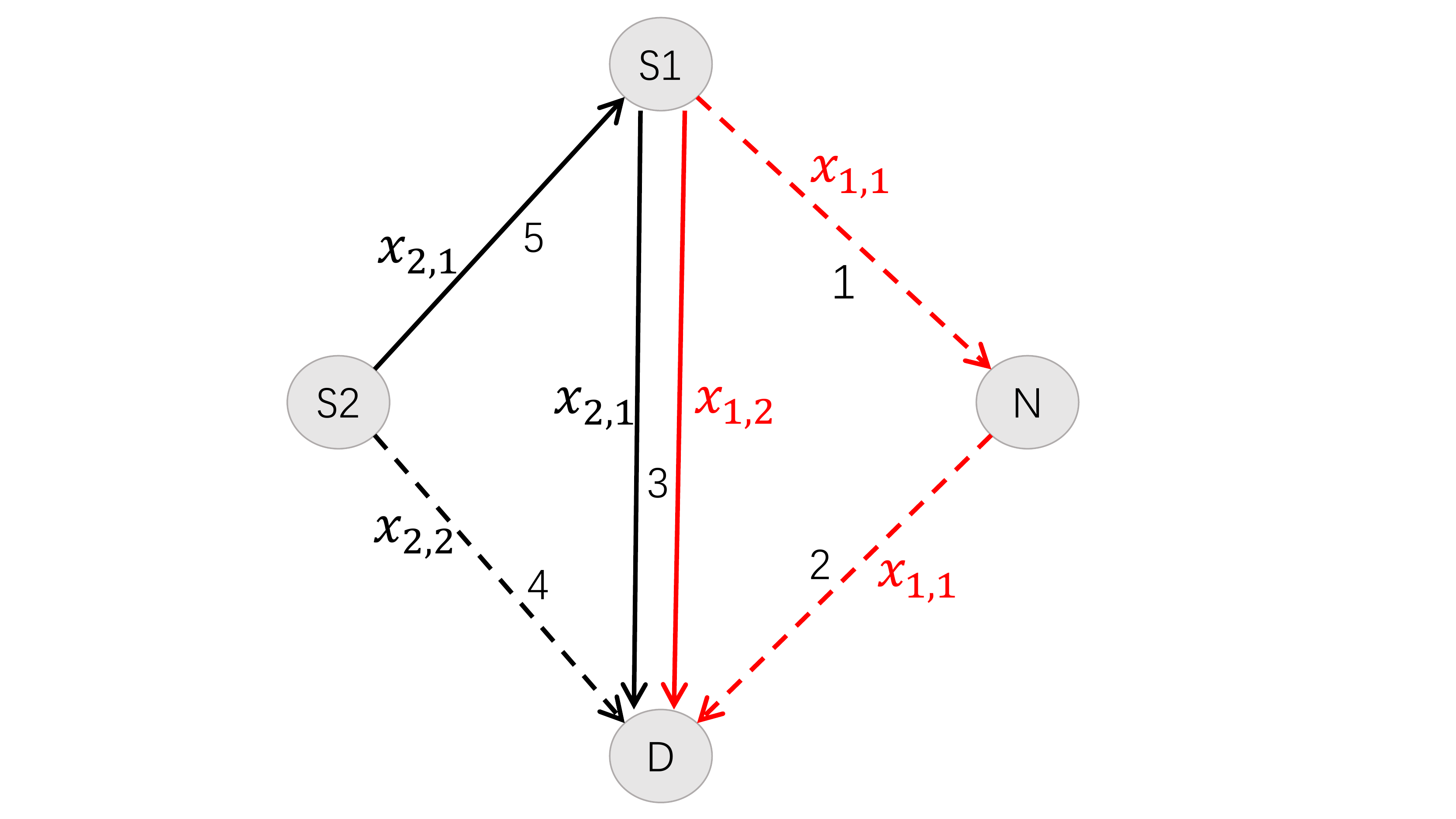} 
	\caption{An illustrative example network with ﬁve links and two users. The ﬁrst user (red line), whose source node is $S1$ and destination node is $D$, has rate $x_{1,1}$ and $x_{1,2}$ on its two paths. The second user (black line), whose source node is $S2$ and destination node is $D$, has rate $x_{2,1}$ and $x_{2,2}$ on its two paths. The numbers next to the lines represent the link indices.}  %图片的名称
	\label{fig: topoexample}   %标签，用作引用
\end{figure}
\figref{fig: topoexample} gives an example of a small network with  five links and two flows (black and red). Each flow has two available paths (solid line and dashed line) and the corresponding routing matrices for the two flows are given by 
\begin{equation*}
\boldsymbol{R}_1={
	\left( \begin{array}{cc}
	1 & 0 \\
	1 & 0 \\
	0 & 1 \\
	0 & 0 \\
	0 & 0
	\end{array} 
	\right )},
\boldsymbol{R}_2={
	\left( \begin{array}{cc}
	0 & 0 \\
	0 & 0 \\
	1 & 0 \\
	0 & 1 \\
	1 & 0
	\end{array} 
	\right )}.
\end{equation*}

\begin{remark}
	In general, since each path only passes a small portion of all links over the entire network, $\boldsymbol{ R}$ is actually very sparse.
\end{remark}

The above network model covers a lot of practical network scenarios. Two examples are illustrated as follows.

\emph{Example 1 (Backbone Network)}: In Backbone network \cite{papagiannaki2003measurement}, Internet service provider acts as the role of delivering various types of application traffic from multiple data sources to terminal users. The network controller will gather information of traffic flows from the users and content providers on network topology information and traffic requirement information. Then it will optimize and allocate the network bandwidth to different applications under certain constraints on physical link limitations and the number of paths requirements.    

\emph{Example 2 (Inter Datacenter Network)}: In Inter-Datacenter Network (IDN) \cite{feng2012airlift}, data are required to be migrated from one data center network to another and different applications have different requirements. For instance, video and backup data require high throughput while finance data require low latency. The IDN controller will optimize the data migration process so as to balance the requirements of different applications while achieving load balancing.

\section{Problem Formulation}
\label{Sec: 3}
In this section, we formulate two types of bandwidth allocation problems. One is the NUM problem \cite{cao2012delay}, in which we use a practical multi-objective
utility function to capture both proportional fairness and low total completion time. The other is a multi-objective bandwidth allocation and path selection problem with path cardinality constraints
(MOPC). We present the two problems in detail as follows.

\subsubsection{NUM}
\label{sec: num}
The NUM
\cite{cao2012delay,kelly1998rate,gupta2016centralized,allybokus2018multi} is a very popular and practical formulation for bandwidth allocation
problems. Let $U_k(\| \boldsymbol{x}_k \|_1)$ be the utility of the $k$-th flow, which is a single-variable function and a mapping from the total allocated bandwidth resource to the utility value. We further assume $U_k(\cdot)$ is increasing, strictly concave and twice continuously differentiable. Some common choices of $U_k(\|
\boldsymbol{x}_k \|_1)$ are $\log(\| \boldsymbol{x}_k \|_1)$ and $-\frac{s_k}{\|
\boldsymbol{x}_k \|_1}$, which represent the proportional fairness and flow completion time, respectively. In the NUM formulation, the goal is to achieve both proportional fairness and low total completion time among all flows subject to link capacity constraints. For a given positive weight $\beta$, the NUM problem is formulated as
\begin{align}
\begin{split}
\min_{\boldsymbol{x}} \quad & -\sum_{k=1}^K U_k( \| \boldsymbol{x}_k \|_1) \\
\mbox{s.t.}\quad
& \boldsymbol{Rx} \leq \boldsymbol{c} \\
& \boldsymbol{x}\geq 0,
\end{split}
\label{eq: NUMproblem}
\end{align}
where 
\begin{align}
U_k( \| \boldsymbol{x}_k \|_1)=
\beta \log(\| \boldsymbol{x}_k \|_1) - \frac{s_k}{\| \boldsymbol{x}_k \|_1}, \quad \forall k.
\label{eq: NUMfunU}
\end{align}
\subsubsection{MOPC}
Another important goal is to balance network load across links so as to avoid network congestion. Let $\boldsymbol{R}[l]$ be the $l$-th row of the routing matrix $\boldsymbol{R}$ as in \eqref{eq: routingmatR}. Therefore, we consider to include the maximal (or worst-case) link utilization ratio, which is defined as $\max_l \frac{\boldsymbol{R}[l]\boldsymbol{x}}{c_l}$, in the objective function to reflect the level of load balancing. In addition, path cardinality constraints are added so that the $k$-th flow can be transmitted through at most $w_k$ paths, where $w_k \leq P_k$. Therefore, for given positive weights $\alpha, \beta$, the MOPC problem is formulated as
\begin{align}
\begin{split}
\min_{\boldsymbol{x}} \quad & \sum_{k=1}^{K} -U_k( \| \boldsymbol{x}_k \|_1) + \alpha \max_l \dfrac{\boldsymbol{R}[l]\boldsymbol{x}}{c_l} \\
\mbox{s.t.}\quad
&\boldsymbol{Rx} \leq \boldsymbol{c} \\
& \| \boldsymbol{x}_k \|_0 \leq w_k,\quad \forall \, k\\
& \boldsymbol{x} \geq \boldsymbol{0}.
\end{split}
\label{eq: MRproblem}
\end{align}

Note that for each $k$, if $w_k = P_k$, in fact there are no path cardinality constraints on the $k$-th flow. Therefore, our problem formulation \eqref{eq: MRproblem} is general enough to cover scenarios with various types of constraints on the number of paths.

\section{ADMM-based Algorithm Design}
\label{Sec: 4}
In this section, we present our proposed ADMM-based algorithm for solving the NUM and MOPC problems. 
\subsection{ADMM for NUM}
The authors in \cite{gupta2016centralized,allybokus2018multi} propose to solve
the NUM problem with ADMM. Compared with the approaches in \cite{gupta2016centralized,allybokus2018multi}, we adopt a different utility function to facilitate our multi-objective optimization problem and also use a different splitting method on the decision variables. Specifically, by introducing the following convex indicator function of link capacity constraints given by
\begin{align}
\label{eq: indicator}
\mathcal{I}_{ \{ \boldsymbol{y}\leq \boldsymbol{c} \} }(\boldsymbol{y})  = 
\begin{cases}
0, & if \; \boldsymbol{y} \leq \boldsymbol{c}, \\
+\infty, & otherwise,
\end{cases}
\end{align} 
we can equivalently transform problem in \eqref{eq: NUMproblem} into
\begin{align}
\begin{split}
\min_{\boldsymbol{x},\boldsymbol{y} } \quad & -\sum_{k=1}^K U_k( \| \boldsymbol{x}_k \|_1 ) + \mathcal{I}_{ \{ \boldsymbol{y}\leq \boldsymbol{c} \} }(\boldsymbol{y})  \\
\mbox{s.t.}\quad
& \boldsymbol{y} = \boldsymbol{Rx} \\
&  \boldsymbol{x} \geq \boldsymbol{0}.
\end{split}
\label{eq: NUMproblemadmm1}
\end{align}
The augmented Lagrangian function of problem in \eqref{eq: NUMproblemadmm1} is then
\begin{align}
L_\rho(\boldsymbol{x},\boldsymbol{y};\boldsymbol{z}) &= -\sum_{k=1}^K U_k( \| \boldsymbol{x}_k \|_1 ) + \mathcal{I}_{ \{ \boldsymbol{y}\leq \boldsymbol{c} \} }(\boldsymbol{y}) \nonumber \\
&+ \boldsymbol{z }^\top( \boldsymbol{y} - \boldsymbol{Rx}) + \dfrac{\rho}{2} \| \boldsymbol{y} - \boldsymbol{Rx}  \|_2^2,
\label{eq: NUMADMMlarg}
\end{align} 
where $\boldsymbol{z}$ is a Lagrangian multiplier vector associated with the constraint $\boldsymbol{y} = \boldsymbol{Rx}$ in \eqref{eq: NUMproblemadmm1}, and $\rho>0$ is a penalty parameter.
The ADMM for problem in \eqref{eq: NUMproblemadmm1} is derived by alternatively minimizing
$L_\rho$ in \eqref{eq: NUMADMMlarg} with respect to
$\boldsymbol{x}$ and $\boldsymbol{y}$ with the other variables fixed. Specifically, the iterative steps are given by
\begin{align}
\boldsymbol{x}^{j+1} &= \mathop{\arg\min}_{ \boldsymbol{x} \geq \boldsymbol{0}  } L_\rho(\boldsymbol{x},\boldsymbol{y}^j;\boldsymbol{z}^j) \nonumber  \\
&= \mathop{\arg\min}_{\boldsymbol{x} \geq \boldsymbol{0} } -\sum_{k=1}^{K} U_k( \| \boldsymbol{x}_k\|_1 )+ \dfrac{\rho}{2} \| \boldsymbol{y}^j - \boldsymbol{Rx} + \frac{\boldsymbol{z}^j}{\rho} \|^2, 
\label{eq: NUMdmmpro11}  \\
\boldsymbol{y}^{j+1} &= \mathop{\arg\min}_{ \boldsymbol{y} } L_\rho(\boldsymbol{x}^{j+1}, \boldsymbol{y};\boldsymbol{z}^j) \nonumber  \\
&= \mathop{\arg\min}_{\boldsymbol{y} \leq \boldsymbol{c} } \dfrac{\rho}{2} \| \boldsymbol{y} - \boldsymbol{Rx}^{j+1} + \frac{\boldsymbol{z}^j}{\rho} \|^2,
\label{eq: NUMdmmpro22} 
\end{align}
where $j$ is the step index, $\boldsymbol{x}^{j+1} = (\boldsymbol{x}_1^{j+1};\boldsymbol{x}_2^{j+1};\cdots;\boldsymbol{x}_K^{j+1})$, and $\boldsymbol{x}_k^{j} = (\boldsymbol{x}_{k,1}^{j},\boldsymbol{x}_{k,2}^{j}, \cdots,\boldsymbol{x}_{k,P_k}^{j})^\top$, $\forall k = 1,2,\cdots, K$. After the update steps of $\boldsymbol{x}$ and $\boldsymbol{y}$ as above,
the update of the multiplier $\boldsymbol{z}$ is given by
\begin{align}
\boldsymbol{z}^{j+1} = \boldsymbol{z}^j + \rho (\boldsymbol{y}^{j+1} - \boldsymbol{Rx}^{j+1}). \label{eq: NUMdmmpro33} 
\end{align}

We next focus on solving the problems in \eqref{eq: NUMdmmpro11} and \eqref{eq: NUMdmmpro22}. For the subproblem in \eqref{eq: NUMdmmpro11}, it is hard to get a closed-form solution since
$\{\boldsymbol{x}_k\}_{k=1,\cdots,K}$ are coupled together due to the quadratic term $\dfrac{\rho}{2} \left\|  \boldsymbol{y}^j - \boldsymbol{Rx} + \frac{\boldsymbol{z}^j}{\rho} \right\| ^2$. To deal with the coupling issue, we consider linearizing the
quadratic term in \eqref{eq: NUMdmmpro11} and adding a proximal term as follows:
\begin{align}
 &\boldsymbol{x}^{j+1}  \nonumber \\
	&= \mathop{\arg\min}_{ \boldsymbol{x} \geq \boldsymbol{0}} \quad  \sum_{k=1}^{K} -U_k( \| \boldsymbol{x}_k \|_1) - \rho(\boldsymbol{q}^j)^\top \boldsymbol{x} + \dfrac{\mu}{2}\| \boldsymbol{x} -\boldsymbol{x}^j  \|^2 \nonumber \\
&= \mathop{\arg\min}_{ \boldsymbol{x} \geq \boldsymbol{0}} \quad  \sum_{k=1}^{K} -U_k( \| \boldsymbol{x}_k \|_1) +  \dfrac{\mu}{2}\| \boldsymbol{x} - \boldsymbol{\nu}^j  \|^2,
\label{eq: NUMx-update12}
\end{align}
where  $\boldsymbol{q}^j = \boldsymbol{R}^\top (\boldsymbol{y}^j - \boldsymbol{Rx}^j +  \boldsymbol{z}^j /\rho)$,  $\boldsymbol{\nu}^j =  \boldsymbol{x}^j + (\rho/\mu)\boldsymbol{q}^j $, and $\mu > \rho \| \boldsymbol{ R} \|_2^2$. Note that $\boldsymbol{\nu}^{j} = (\boldsymbol{\nu}_1^{j};\boldsymbol{\nu}_2^{j};\cdots;\boldsymbol{\nu}_K^{j})$, and for $\forall k = 1,2,\cdots, K$, $\boldsymbol{\nu}_k^{j} = (\boldsymbol{\nu}_{k,1}^{j},\boldsymbol{\nu}_{k,2}^{j}, \cdots,\boldsymbol{\nu}_{k,P_k}^{j})^\top$. 
Therefore based on \eqref{eq: NUMx-update12}, the problem in the right hand side can be decomposed into several subproblems, where each subproblem only contains one block decision variable $\boldsymbol{x}_k$ as follows:
\begin{align}
\boldsymbol{x}^{j+1}_k= \mathop{\arg\min}_{ \boldsymbol{x}_k \geq \boldsymbol{0}} \quad  -U_k( \| \boldsymbol{x}_k \|_1) +  \dfrac{\mu}{2}\| \boldsymbol{x}_k - \boldsymbol{\nu}_k^j  \|^2, \forall k.
\label{eq: NUMx-subproblemblock2}
\end{align}
The solution to problem in \eqref{eq: NUMx-subproblemblock2} is then given in the following Lemma \ref{lemma: xupdateinfty}. 
\begin{lemma}
	The solution of problem in \eqref{eq: NUMx-subproblemblock2} is given by
	\begin{align*}
	\begin{cases}
	&x_{k,i}^{j+1} = \max(0, \nu_{k,i}^j + \zeta), \quad \forall i,\\
	& \mu  \zeta = U'_k \left( \sum_{i=1}^{P_k} \max(0, \nu_{k,i}^j +  \zeta) \right),
	\end{cases}
	\end{align*}
	where $ \zeta$ is some positive number. Without loss of generality, suppose the elements of $\boldsymbol{\nu}_k^j = (\nu_{k,1}^j, \nu_{k,2}^j,\cdots,\nu_{k,P_k}^j)^\top$ are in descending order. Then if there is some smallest index $i'$ $(i' \in {2,\cdots, P_k}) $ such that
	\begin{align*}
	\mu \zeta \geq U'_k \left( \sum_{i=1}^{P_k} \max(0, \nu_{k,i}^j + \zeta) \right),
	\end{align*} 
	where $\zeta = -\nu_{k,i'}^j $,
	we can obtain $\zeta$ through finding the maximal root of a single-variable cubic equation
	\begin{align}
	\mu \zeta \left(\sum_{i=1}^{i'-1} (\nu^j_{k,i}+\zeta) \right)^2 = \beta \left(\sum_{i=1}^{i'-1} (\nu^j_{k,i}+\zeta) \right) + s_k,
	\label{eq: numcubic1}
	\end{align}
	otherwise we can obtain $\zeta$ through finding the maximal root of a single-variable cubic equation
	\begin{align}
	\mu \zeta \left(\sum_{i=1}^{P_k} (\nu^j_{k,i}+\zeta) \right)^2 = \beta \left(\sum_{i=1}^{P_k} (\nu^j_{k,i}+\zeta) \right) + s_k.
	\label{eq: numcubic2}
	\end{align}
	\label{lemma: xupdateinfty}
\end{lemma}
\begin{proof}
	See Appendix \ref{Prooflemma2}.
\end{proof}
For the subproblem in \eqref{eq: NUMdmmpro22}, the solution can be easily obtained by performing a projection operation. That is, for $l = 1,2,\cdots,L$,
\begin{equation}
y_l^{j+1}=\left\{
\begin{array}{ll}
\theta_l^j, &   if \quad \theta_l^j \leq c_l, \\
c_l, &   otherwise,
\end{array} \right.
\label{eq: NUMysub}
\end{equation}
where $\boldsymbol{ \theta}^j = \boldsymbol{ R}\boldsymbol{ x}^{j+1} - \boldsymbol{ z}^j/\rho$.
		\begin{figure}[]%%图
	\centering  %插入的图片居中表示
	\includegraphics[scale=0.28]{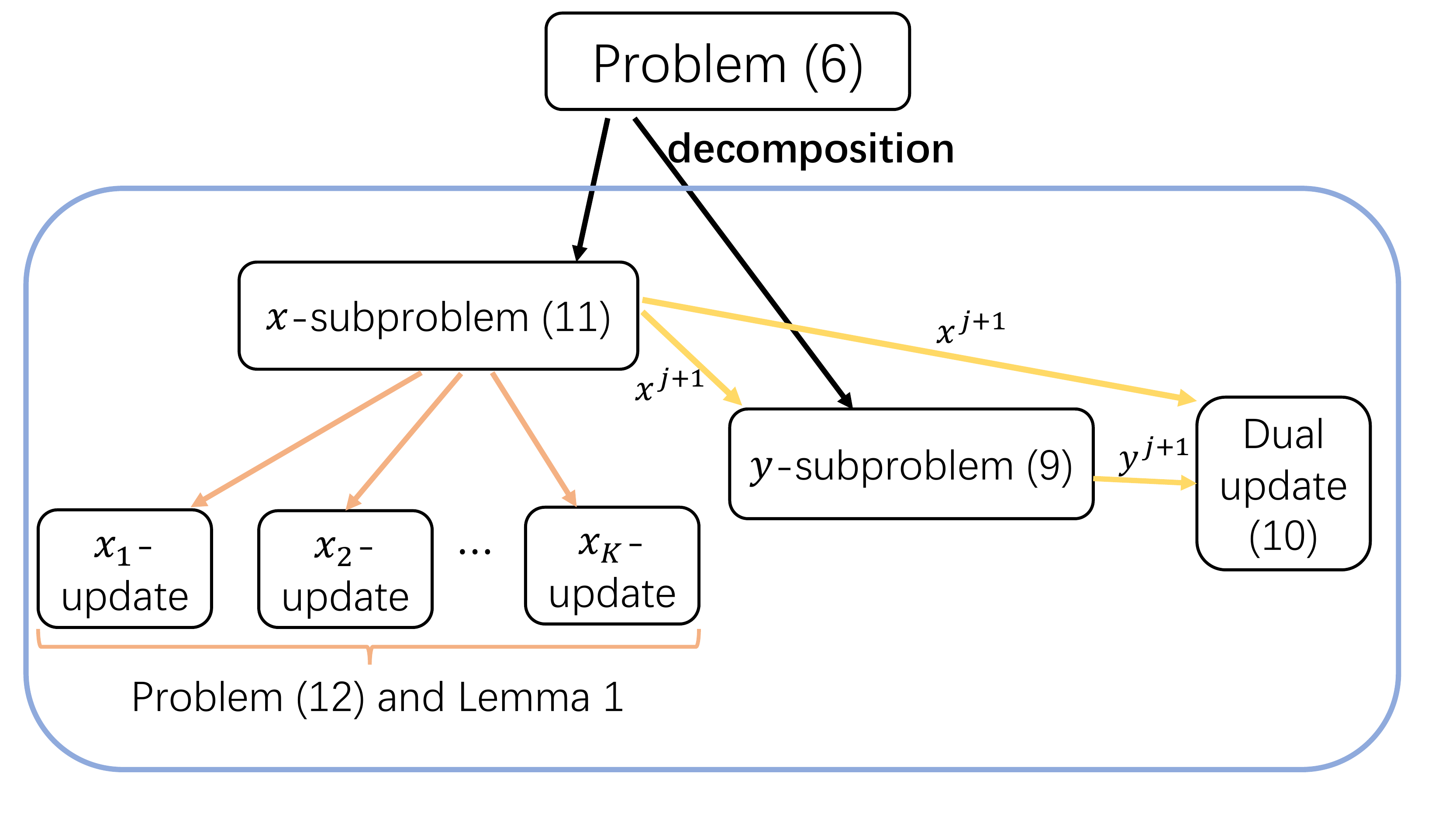} 
	\caption{Illustration of the decomposition procedure of our ADMM for NUM.}  %图片的名称
	\label{fig: numalgprocedure}   %标签，用作引用
\end{figure}

  The overall decomposition procedure using the ADMM approach on the NUM is outlined in \figref{fig: numalgprocedure}. Adopting the stopping criteria in the Section 3.3.1 in \cite{boyd2011distributed}, the complete procedure of our ADMM for NUM is summarized in Algorithm \ref{alg: ADMM}.
\begin{algorithm}[htbp]
	\caption{ADMM for the NUM} 
	\label{alg: ADMM} 
	\begin{algorithmic}[1] %这个1 表示每一行都显示数字
		\REQUIRE %算法的输入参数：Input
		A randomly initialized point $(\boldsymbol{x}^0, \boldsymbol{y}^0,\boldsymbol{z}^0)$,\\
		~~~~\; set $j=0$, $\rho > 0$, and $\mu > \rho \|\boldsymbol{R} \|^2$.
		\ENSURE %算法的输出：Output
		 An optimal bandwidth allocation and path selection vector $\boldsymbol{x}^*$.
		
		\WHILE{The stopping criteria (see Section 3.3.1 in \cite{boyd2011distributed}) are not satisfied}

		\STATE(Solve the $\boldsymbol{x}$-minimization problem in \eqref{eq: NUMx-subproblemblock2}  according to Lemma \ref{lemma: xupdateinfty})
		\FOR{$k = 1 \to K$}
		\STATE Sort the elements of $\boldsymbol{\nu}^j_k$ in descending order.
		\STATE Try every positive $\zeta=-\nu^j_{k,i} (i=2,\cdots,P_k)$ to determine the corresponding interval in which $\zeta$ satisfies $\mu \zeta = U'_k \left( \sum_{i=1}^{P_k} \max(0, \nu^j_{k,i} + \zeta) \right) $.
		\STATE Solve the single-variable cubic equation in \eqref{eq: numcubic1} or \eqref{eq: numcubic2}.
		\STATE Let $x_{k,i}^{j+1} = \max(0, \nu^j_{k,i} + \zeta), \quad \forall i$.
		\ENDFOR 
		\STATE (Solve the $\boldsymbol{y}$-minimization problem \eqref{eq: NUMdmmpro22} with \eqref{eq: NUMysub})
		\STATE Compute $\boldsymbol{\theta}^j = \boldsymbol{ R}\boldsymbol{x}^{j+1}-\boldsymbol{z}^j/\rho$.
		\STATE For each $l=1,\cdots,L$, perform \begin{equation*}
		y_l^{j+1}=\left\{
		\begin{array}{ll}
		\theta^j_l, &   if \quad \theta^j_l \leq c_l, \\
		c_l, &   otherwise.
		\end{array} \right.
		\end{equation*}
		\STATE Update the multiplier $\boldsymbol{z} $ according to \eqref{eq: NUMdmmpro33}.
				\STATE Set $j = j +1$.
		\ENDWHILE
	\end{algorithmic}
\end{algorithm}

\subsection{ADMM for MOPC}
\label{sec: centralizedalg}
In this subsection, we solve the MPOC using the ADMM approach. We first transform \eqref{eq: MRproblem} into the following equivalent form:
\begin{align}
\begin{split}
\min_{\boldsymbol{x} \in \mathcal{X},\boldsymbol{y}\in \mathcal{Y}} \quad &  \sum_{k=1}^{K} -U_k( \| \boldsymbol{x}_k \|_1) + \alpha \max_l \dfrac{y_l}{c_l} \\
\mbox{s.t.}\quad
& \boldsymbol{y} = \boldsymbol{Rx},
\end{split}
\label{eq: problem2}
\end{align}
where $\mathcal{X} = \{\boldsymbol{ x}| \| \boldsymbol{x}_k  \|_0\leq w_k, \, \forall k; \boldsymbol{x} \geq \boldsymbol{0} \}$, and $\mathcal{Y} = \{\boldsymbol{ y}|  \boldsymbol{ 0} \leq \boldsymbol{y} \leq \boldsymbol{c} \}$.
The augmented Lagrangian function of problem in \eqref{eq: problem2} is
\begin{align}
L_\rho(\boldsymbol{x},\boldsymbol{y};\boldsymbol{z}) & = \sum_{k=1}^{K} -U_k( \| \boldsymbol{x}_k \|_1) + \alpha\max_l \dfrac{y_l}{c_l}   \nonumber  \\
& + \boldsymbol{z}^\top( \boldsymbol{y} - \boldsymbol{Rx}) + \dfrac{\rho}{2} \| \boldsymbol{y} - \boldsymbol{Rx}  \|_2^2,
\label{eq: lagsrp}
\end{align} 
where $\boldsymbol{z}$ is a Lagrangian multiplier vector  associated with the constraint $\boldsymbol{y} = \boldsymbol{Rx}$ in \eqref{eq: problem2} and $\rho>0$ is a penalty parameter.

 Similar to the ADMM for the NUM, the ADMM for the MOPC problem in \eqref{eq: problem2} performs minimization of $L_\rho$ in \eqref{eq: lagsrp} with respect
to $\boldsymbol{x}$ and $\boldsymbol{y}$ alternatively followed by the update of $\boldsymbol{z}$. That is, at iteration $j$, the following updates of variables are performed:
\begin{align}
\boldsymbol{x}^{j+1} &= \mathop{\arg\min}_{ \boldsymbol{x} \in \mathcal{X} } L_\rho(\boldsymbol{x},\boldsymbol{y}^j ; \boldsymbol{z}^j) \nonumber  \\
&{ = \mathop{\arg\min}_{\boldsymbol{x} \in \mathcal{X}} \sum_{k=1}^{K} -U_k( \| \boldsymbol{x}_k \|_1)+ \dfrac{\rho}{2} \| \boldsymbol{y}^j - \boldsymbol{Rx} + \frac{\boldsymbol{z}^j}{\rho} \|^2,}
\label{eq: scaleadmmpro1}\\ 
%\end{align}
%\begin{align}
\boldsymbol{y}^{j+1} &= \mathop{\arg\min}_{\boldsymbol{y} \in \mathcal{Y} } L_\rho(\boldsymbol{x}^{j+1},\boldsymbol{y};\boldsymbol{z}^j)  \nonumber\\
& = \mathop{\arg\min}_{\boldsymbol{y} \in \mathcal{Y} } \alpha \max_l \dfrac{y_l}{c_l}+\dfrac{\rho}{2} \| \boldsymbol{y} - \boldsymbol{Rx}^{j+1} + \frac{\boldsymbol{z}^j}{\rho} \|^2,
\label{eq: scaleadmmpro2}\\
%\end{align}
%\begin{align}
\boldsymbol{z}^{j+1} &= \boldsymbol{z}^j +  \rho (\boldsymbol{y}^{j+1} - \boldsymbol{Rx}^{j+1}). \label{eq: scaleadmmpro3}
\end{align}

We next focus on solving the subproblem in \eqref{eq:
scaleadmmpro1}. Similar to the technique in solving \eqref{eq: NUMdmmpro11}, we also consider linearizing
the quadratic term in \eqref{eq: scaleadmmpro1} and adding a proximal term:
\begin{align}
 & \boldsymbol{x}^{j+1} \nonumber \\
 &= \mathop{\arg\min}_{ \boldsymbol{ x} \in \mathcal{X}} \quad  \sum_{k=1}^{K} -U_k( \| \boldsymbol{x}_k \|_1) - \rho(\boldsymbol{q}^j)^\top \boldsymbol{x} + \dfrac{\mu}{2}\| \boldsymbol{x} -\boldsymbol{x}^j  \|^2 \nonumber \\
&= \mathop{\arg\min}_{ \boldsymbol{ x} \in \mathcal{X}} \quad  \sum_{k=1}^{K} -U_k( \| \boldsymbol{x}_k \|_1) +  \dfrac{\mu}{2}\| \boldsymbol{x} - \boldsymbol{\nu}^j  \|^2,
\label{x-update1}
\end{align}
where  $\boldsymbol{q}^j = \boldsymbol{R}^\top (\boldsymbol{y}^j - \boldsymbol{Rx}^j +  \boldsymbol{z}^j /\rho), \boldsymbol{\nu}^j =  \boldsymbol{x}^j + (\rho/\mu) \boldsymbol{q}^j $, and $\mu > \rho \|\boldsymbol{ R} \|^2$. 
We decompose \eqref{x-update1} into several subproblems, where each subproblem is associated with only one block $\boldsymbol{x}_k$,
\begin{align}
	\boldsymbol{x}^{j+1}_k= \mathop{\arg\min}_{ \| \boldsymbol{x}_k  \|_0\leq w_k, \boldsymbol{x}_k \geq \boldsymbol{0}}   -U_k( \| \boldsymbol{x}_k \|_1) +  \dfrac{\mu}{2}\| \boldsymbol{x}_k - \boldsymbol{\nu}_k^j  \|^2, \forall k.
	\label{x-subproblemblock}
\end{align}
Leveraging the strict concaveness and monotonicity of the utility function $U_k(\cdot)$ in $(0,+\infty)$, we obtain the solution to \eqref{x-subproblemblock} in the following two lemmas as well as Lemma \ref{lemma: xupdateinfty}.

\begin{lemma}
	
If $w_k=1$, for block
	$\boldsymbol{x}_k$, we denote $i\, (i\in \{1,2,\cdots, P_k\} )$ as the index
	where $\nu_{k,i}^j$ attains the maximal value, then only the $i$-th element of $\boldsymbol{x}_k^{j+1}$, i.e., $x_{k,i}^{j+1}$ is non-zero and its value is the maximal root of a single-variable cubic equation 
	\begin{align}
	\mu \left(x_{k,i}^{j+1}\right)^3 - \mu \nu_{k,i}^j \left(x_{k,i}^{j+1}\right)^2 - \beta \left(x_{k,i}^{j+1}\right) - s_k  =0,\, \forall k,
	\label{eq: lemma2cubic}
	\end{align}
which always has at least one positive root.
	\label{lemma: xupdate1}
\end{lemma}
\begin{proof}
	See Appendix \ref{Prooflemma1}.
\end{proof}

\begin{lemma}

If $ 1<w_k <P_k$, without loss of generality, suppose the elements of $\boldsymbol{\nu}_k^j$ are in descending order. Then one of its solutions (there may exist ``many'' solutions because $\boldsymbol{\nu}_k^j$ may have equal elements) satisfies
	\begin{subequations}
		\begin{align}
		&x_{k,i}^{j+1} = \max(0, \nu^j_{k,i} + \zeta), \quad i=1,\cdots,w_k, \label{eq: lemmaxs1} \\
		& \mu \zeta = U'_k \left( \sum_{i=1}^{w_k} \max(0, \nu^j_{k,i} + \zeta) \right), \label{eq: lemmaxs2} \\
		&x_{k,i}^{j+1} = 0, \quad i > w_k, \label{eq: lemmaxs3}
		\end{align}
	\end{subequations}
Where $\zeta$ is some positive number. If there is some smallest index $i'$ $(i' \in {2,\cdots, w_k}) $ such that
	\begin{align*}
	\mu \zeta \geq U'_k \left( \sum_{i=1}^{w_k} \max(0, \nu_{k,i}^j + \zeta) \right),
	\end{align*} 
	where $\zeta = -\nu_{k,i'}^j $,
	we can obtain $\zeta$ through finding the maximal root of a single-variable cubic equation
	\begin{align}
	\mu \zeta \left(\sum_{i=1}^{i'-1} (\nu^j_{k,i}+\zeta) \right)^2 = \beta \left(\sum_{i=1}^{i'-1} (\nu^j_{k,i}+\zeta) \right) + s_k,
	\label{eq: numcubic3}
	\end{align}
	otherwise we can obtain $\zeta$ through finding the maximal root of a single-variable cubic equation
	\begin{align}
	\mu \zeta \left(\sum_{i=1}^{w_k} (\nu^j_{k,i}+\zeta) \right)^2 = \beta \left(\sum_{i=1}^{w_k} (\nu^j_{k,i}+\zeta) \right) + s_k.
	\label{eq: numcubic4}
	\end{align}
	\label{lemma: xupdate2}
\end{lemma}
\begin{proof}
	See Appendix \ref{Prooflemma3}.
\end{proof}
\begin{remark}
	It is worth noting that the $\boldsymbol{x}$-subproblem in \eqref{x-update1} can be solved in a parallel manner because it can be decomposed into several independent subproblems in \eqref{x-subproblemblock}, each of which involves only one block variable $\boldsymbol{x}_k$.
\end{remark}
%\begin{remark}
%	We can also directly solve the $\boldsymbol{x}$-subproblem in \eqref{eq:
%    scaleadmmpro1} instead of solving the approximated problem in \eqref{x-update1} under certain circumstances. More specifically, when $w_k \geq P_k$, the Parallel Successive Convex Approximate (PSCA) algorithm \cite{hong2015unified,razaviyayn2014parallel} can be applied and the corresponding subproblems are:
%	\begin{align}
%			\min_{ \boldsymbol{x}_r \geq \boldsymbol{0} }   -U_r( \| \boldsymbol{x}_r \|_1) + \tilde{q}_r(\boldsymbol{x}_r), \quad r = 1,\cdots,K,
%			\label{eq: PSCA}
%	\end{align}
%	where $\tilde{q}_r(\boldsymbol{x}_r)$ is an approximate upper bound of $ \dfrac{\rho}{2} \| \boldsymbol{y}^j - \boldsymbol{ R}_r \boldsymbol{x}_r - \sum_{k\neq r}\boldsymbol{ R}_k \boldsymbol{x}_k + {\boldsymbol{z}^j}/{\rho} \|^2 $. The solution to \eqref{eq: PSCA} can be obtained from Lemma \ref{lemma: xupdateinfty}.
%\end{remark}

%\begin{remark}
%	Note that our utility function can also be nonconvex such as the Sigmoid function, see \cite{gotoh2018dc,fazel2005network}
%\end{remark}

Note that it is quite hard to obtain a closed-form solution to the subproblem in \eqref{eq: scaleadmmpro2} since it is a non-smooth problem with linear constraints. By introducing a new variable $t = \max_l \dfrac{y_l}{c_l}$, we transform it into an equivalent quadratic program 
\begin{align}
\begin{split}
\min_{t, \boldsymbol{y}} \quad & \phi(t,\boldsymbol{y}) =  \alpha t + \dfrac{\rho}{2}\| \boldsymbol{y} - \boldsymbol{Rx}^{j+1} + \boldsymbol{z}^j/\rho    \|^2 \\
\mbox{s.t.}\quad
& \boldsymbol{y}\leq t\boldsymbol{c}  \\
& \boldsymbol{y}\geq  \boldsymbol{0} \\
&  \boldsymbol{y}\leq \boldsymbol{c}.
\end{split}
\label{eq: y-updatesub}
\end{align}
The solution to problem \eqref{eq: y-updatesub} (equivalently problem \eqref{eq: scaleadmmpro2}) is given as follows:
\begin{lemma}
	Consider problem \eqref{eq: y-updatesub}. Let $t^*$ be the optimal solution to
	\begin{align}
	\min_t  \Phi(t) = \phi(t,\boldsymbol{y}(t))  \quad s.t. \quad 0 \leq t \leq 1,
	\label{eq: singlevar}
	\end{align}
	where $\boldsymbol{y}(t) = ( y_1(t), \cdots, y_L(t))^\top$ and $y_l(t), l=1,\cdots,L$ are defined by
	\begin{equation}
	y_l(t)=\left\{
	\begin{array}{ll}
	\theta^j_l, &   if \quad 0 \leq \theta^j _l \leq  \min(c_l, tc_l), \\
	\min(c_l, tc_l), &    if \quad  \theta^j _l >  \min(c_l, tc_l) ,\\
	0, &   otherwise,
	\end{array} \right.
	\label{eq: ysubsolution}
	\end{equation}
	where $\boldsymbol{\theta}^j = \boldsymbol{Rx}^{j+1} - {\boldsymbol{z}^j}/\rho$. 
	Then $(t^*, \boldsymbol{y}(t^*))$ is the optimal solution to problem  \eqref{eq: y-updatesub} and $\boldsymbol{y}(t^*)$ is the optimal solution to problem \eqref{eq: scaleadmmpro2}. Solving problem \eqref{eq: singlevar} is to minimize a single-variable convex function on a closed interval, which can be solved by golden section search and parabolic interpolation (e.g., the function “fminbnd” \cite{forsythe1977computer,brent2013algorithms} in MATLAB). 
	\label{lemma: y-subproblem}
\end{lemma}
\begin{proof}
	See Appendix \ref{Prooflemma4}.
\end{proof}

		\begin{figure}[]%%图
	\centering  %插入的图片居中表示
	\includegraphics[scale=0.28]{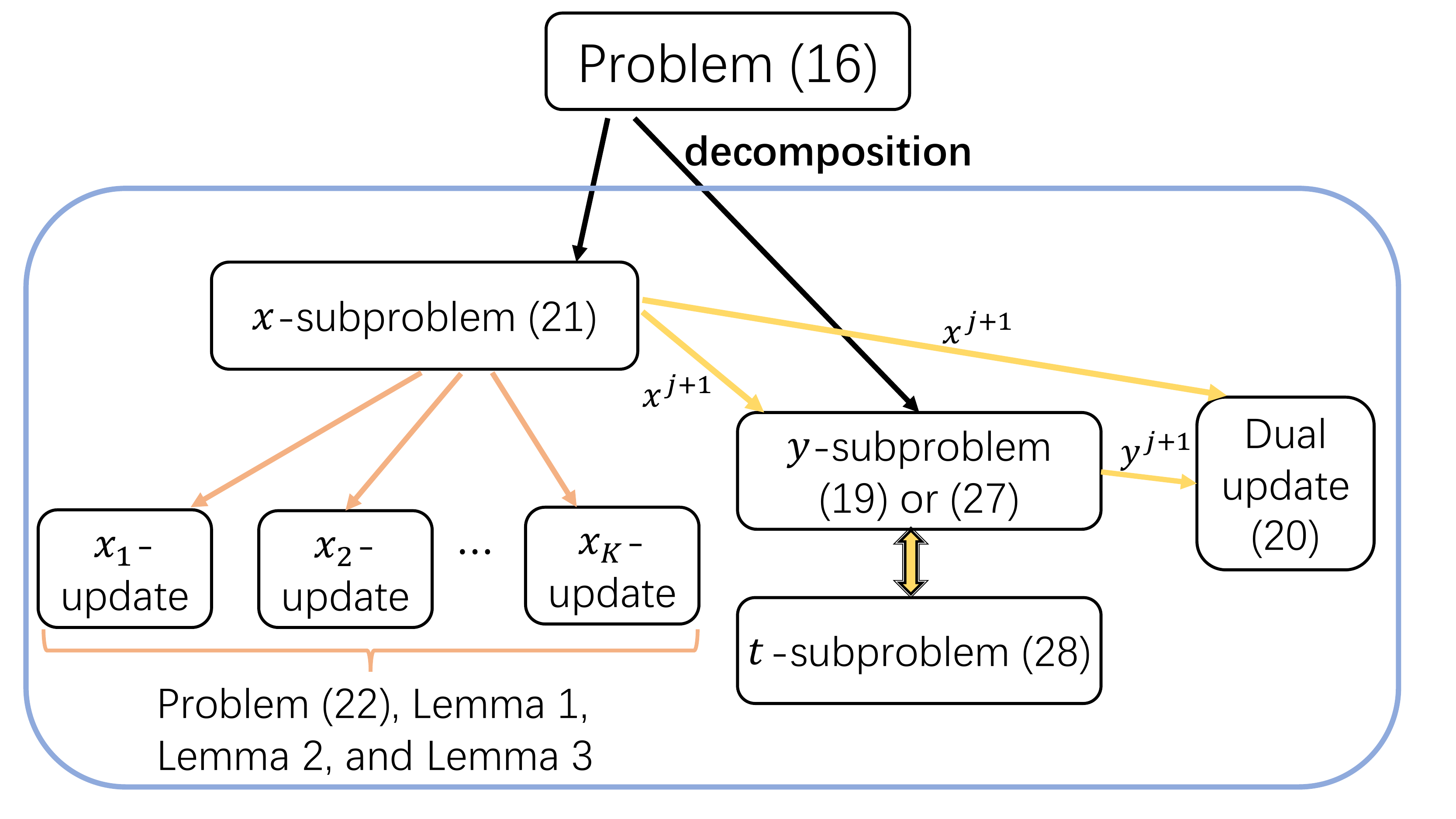} 
	\caption{Illustration of the decomposition procedure of our ADMM for MOPC.}  %图片的名称
	\label{fig: algprocedure}   %标签，用作引用
\end{figure}
  
   \figref{fig: algprocedure} gives the overall decomposition procedure using the ADMM approach on the MOPC. The complete procedure of ADMM for MOPC is summarized in Algorithm \ref{alg: ADMMforsrp}.
\begin{algorithm}[htbp]
	\caption{ADMM for the MOPC} 
	\label{alg: ADMMforsrp} 
	\begin{algorithmic}[1] %这个1 表示每一行都显示数字
		\REQUIRE  %算法的输入参数：Input
		A randomly initialized point $(\boldsymbol{x}^0,  \boldsymbol{y}^0,\boldsymbol{z}^0)$, $w_k\, (k=1,2,\cdots,K)$, $j=0$, $\rho > 0$, and $\mu > \rho \|\boldsymbol{R} \|^2.$
		\ENSURE  %算法的输出：Output
		The optimal bandwidth allocation and path selection vector $\boldsymbol{x}^*$ and corresponding $\boldsymbol{y}^*$.
		
		\WHILE{The stopping criteria (stated in Section \ref{subsec: ps}) are not satisfied}
		\STATE (Solve the $\boldsymbol{x}$-minimization problem in \eqref{x-subproblemblock})
		\FOR{$k = 1 \to K$}
		\IF{$w_k=1$}
		\STATE Solve the single-variable cubic equation in \eqref{eq: lemma2cubic} to get $\boldsymbol{x}^{j+1}_k$.
		\ELSIF{$1 < w_k < P_k$}
		\STATE Sort the elements of $\boldsymbol{\nu}^j_k$ in descending order.
		\STATE Try every positive $\zeta=-\nu^j_{k,i} (i=2,\cdots,w_k)$ to determine the corresponding interval in which $\zeta$ satisfies $\mu \zeta = U'_k \left( \sum_{i=1}^{w_k} \max(0, \nu^j_{k,i} + \zeta) \right) $.
\STATE Solve the single-variable cubic equation in \eqref{eq: numcubic3} or \eqref{eq: numcubic4}.
\STATE Get $x_{k,i}^{j+1}$ according to \eqref{eq: lemmaxs1} and \eqref{eq: lemmaxs3}. 
		\ELSE 
		\STATE Sort the elements of $\boldsymbol{\nu}^j_k$ in descending order.
\STATE Try every positive $\zeta=-\nu^j_{k,i} (i=2,\cdots,P_k)$ to determine the corresponding interval in which $\zeta$ satisfies $\mu \zeta = U'_k \left( \sum_{i=1}^{P_k} \max(0, \nu^j_{k,i} + \zeta) \right) $.
\STATE Solve the single-variable cubic equation in \eqref{eq: numcubic1} or \eqref{eq: numcubic2}.
\STATE Let $x_{k,i}^{j+1} = \max(0, \nu^j_{k,i} + \zeta), \quad \forall i$.
		\ENDIF
		\ENDFOR 
   		\STATE (Solve the $\boldsymbol{y}$-minimization problem \eqref{eq: scaleadmmpro2} according to Lemma \ref{lemma: y-subproblem})
		\STATE Solve the single-variable optimization problem in \eqref{eq: singlevar} using ``fminbnd'' in MATLAB to get the optimal $t$.
		\STATE Get $\boldsymbol{y}^{j+1}$ with \eqref{eq: ysubsolution}.
		\STATE Update the multiplier $\boldsymbol{z} $ according to \eqref{eq: scaleadmmpro3}.
				\STATE Set $j = j +1$.
		\ENDWHILE
	\end{algorithmic}
\end{algorithm}

\section{Convergence Analysis}
\label{Sec: 5}
In this section, we analyze the convergence property of the proposed algorithm in Algorithm \ref{alg: ADMMforsrp}. As for Algorithm \ref{alg: ADMM}, since the NUM in \eqref{eq: NUMproblem} is a convex problem, the convergence of Algorithm \ref{alg: ADMM} has already been well studied in \cite{boyd2011distributed} and \cite{chen2017note}.
For simplicity of notation, we rewrite problem \eqref{eq: problem2} as
\begin{align}
\begin{split}
\min_{\boldsymbol{x},\boldsymbol{y}} \quad & f(\boldsymbol{x}) + g(\boldsymbol{y}) + \mathcal{I}_{ \mathcal{Y} }(\boldsymbol{y}) + \mathcal{I}_{\mathcal{X} }(\boldsymbol{x}) \\
\mbox{s.t.}\quad
& \boldsymbol{y} = \boldsymbol{Rx},    
\end{split}
\label{eq: problemKKT}
\end{align}
where 
\begin{align*}
f(\boldsymbol{x})=  \sum_{k=1}^{K} -U_k( \sum_{i=1}^{P_k} x_{k,i}),
\end{align*} 
\begin{align*}
g(\boldsymbol{y}) = \alpha \max_l \dfrac{y_l}{c_l},
\end{align*}
and $\mathcal{I}_{ \mathcal{Y} }(\boldsymbol{y})$ as well as $\mathcal{I}_{\mathcal{X} }(\boldsymbol{x})$ are indicator functions defined in \eqref{eq: indicator}.
The augmented Lagrangian function can be rewritten as
\begin{align}
	L_\rho(\boldsymbol{x},\boldsymbol{y};\boldsymbol{z})  = f(\boldsymbol{x}) + g(\boldsymbol{y}) + \mathcal{I}_{ \mathcal{X} }(\boldsymbol{x})+ \mathcal{I}_{ \mathcal{Y}}(\boldsymbol{y})   \nonumber \\
	+ \boldsymbol{z}^\top ( \boldsymbol{y}-\boldsymbol{R}\boldsymbol{x} ) + \frac{\rho}{2} \| \boldsymbol{y}-\boldsymbol{R}\boldsymbol{x} \|^2. \label{eq: algconv}
\end{align}
Next, we show the boundedness of the augmented Lagrangian function, which plays an important role in our convergence analysis.
\begin{proposition}
	\label{propo: alg2albound}
	Assume that $\{\boldsymbol{z}^j \} $ is bounded, then the sequence of the augmented Lagrangian 	$\{L_\rho(\boldsymbol{x}^j,\boldsymbol{y}^j;\boldsymbol{z}^j)\}_{j=1}^{+\infty}$ is bounded below.
\end{proposition}
\begin{proof}
	See Appendix \ref{Proofpropo1}.
\end{proof}
Even through the function $\mathcal{I}_{\mathcal{X} }(\boldsymbol{x})$ is non-smooth and non-convex, we  know from \cite{attouch2013convergence} and \cite{bolte2014proximal} that there exists \textit{limiting subdifferential} or simply the \emph{subdifferential} for it, written as $\partial \mathcal{I}_{\mathcal{X} }(\boldsymbol{x})$.

In the rest of this paper, for simplicity, the following notations are used for the successive errors of the iterative sequence:
\begin{align*}
E_{\boldsymbol{x}}^{j+1}\coloneqq \boldsymbol{x}^{j+1}- \boldsymbol{x}^j,  E_{\boldsymbol{y}}^{j+1}\coloneqq \boldsymbol{y}^{j+1}- \boldsymbol{y}^j, E_{\boldsymbol{z}}^{j+1}\coloneqq \boldsymbol{z}^{j+1}- \boldsymbol{z}^j.
\end{align*}

Before proceeding with our proof, we first show the KKT conditions of problem in
\eqref{eq: problemKKT}  and the optimality conditions of the subproblems of ADMM. The KKT conditions of problem \eqref{eq: problemKKT} are there exists $ \hat{\boldsymbol{z}}^* \in \mathcal{R}^L$, such that
\begin{subequations}
\begin{align}
	&\boldsymbol{ 0} \in \nabla f(\boldsymbol{x}^*) + \partial \mathcal{I}_{\mathcal{X}}(\boldsymbol{x}^*)-\boldsymbol{ R}^\top \hat{\boldsymbol{ z}}^*, \label{eq: KKTx} \\
	&\boldsymbol{ 0} \in \partial g(\boldsymbol{y}^*) + \partial \mathcal{I}_{\mathcal{Y}}(\boldsymbol{y}^*) + \hat{\boldsymbol{z}}^*, \label{eq: KKTy}\\
	&\boldsymbol{y}^* = \boldsymbol{R}\boldsymbol{x}^*. \label{eq: KKTcon}
\end{align}
\end{subequations}
The first-order optimality condition of $\boldsymbol{x}$-subproblem in \eqref{x-update1} is
\begin{align}
	&  \exists \, \tilde{\boldsymbol{x}}^{j+1} \in \partial \mathcal{I}_{ \mathcal{X} }(\boldsymbol{ x}^{j+1}), \text{such that} \nonumber \\
	& \nabla f(\boldsymbol{x}^{j+1}) + \tilde{\boldsymbol{x}}^{j+1} 
	- \rho\boldsymbol{q}^j + \mu (\boldsymbol{x}^{j+1} - \boldsymbol{x}^j)= \boldsymbol{0} ,
	\label{eq: optx}
\end{align}

The first-order optimality condition of $\boldsymbol{y}$-subproblem in \eqref{eq: scaleadmmpro2} is
\begin{align}
&  \exists \, \tilde{\boldsymbol{g}}^{j+1} \in \partial g(\boldsymbol{y}^{j+1}), \text{and}\; \exists\, \tilde{\boldsymbol{y}}^{j+1} \in \partial \mathcal{I}_{ \mathcal{Y} }(\boldsymbol{y}^{j+1}), \text{such that} \nonumber \\
& \tilde{\boldsymbol{g}}^{j+1} + \tilde{\boldsymbol{y}}^{j+1} + \rho ( \boldsymbol{y}^{j+1} - \boldsymbol{Rx}^{j+1} + \boldsymbol{z}^j/\rho) =\boldsymbol{0},
\label{eq: opty1}
\end{align}
which is equivalent to
\begin{align}
	\tilde{\boldsymbol{g}}^{j+1} + \tilde{\boldsymbol{y}}^{j+1} +  \boldsymbol{z}^{j+1} =\boldsymbol{0}.
	\label{eq: opty2}
\end{align}
Next, we estimate the differences of the augmented Lagrangian function values between two successive iterations of Algorithm \ref{alg: ADMMforsrp} to show the sufficient decrease.
\begin{lemma}
	\label{lemma: weakdc}
	For all $j \geq 0$, we have 
	\begin{align}
	& L_\rho(\boldsymbol{x}^{j},\boldsymbol{y}^j;\boldsymbol{z}^j) - L_\rho(\boldsymbol{x}^{j+1},\boldsymbol{y}^{j+1};\boldsymbol{z}^{j+1} )  + \frac{1}{\rho} \| E_{\boldsymbol{z}}^{j+1}  \|^2  \nonumber \\
	&\geq  \frac{\mu - \rho \| \boldsymbol{R} \|_2^2}{2} \| E_{\boldsymbol{x}}^{j+1} \|^2 + \frac{\rho}{2} \| E_{\boldsymbol{y}}^{j+1} \|^2.
	\label{eq: weakdc}
	\end{align}
\end{lemma}
\begin{proof}
	See Appendix \ref{Prooflemma5}.
\end{proof}

At this point, we readily have the following theorem regarding the optimality of the output of Algorithm \ref{alg: ADMMforsrp}.
\begin{theorem} \label{thm:conv}
	Let $ \{ \boldsymbol{x}^j,\boldsymbol{y}^j,\boldsymbol{z}^j \}_{j=1}^{+\infty}$ be a sequence generated by Algorithm \ref{alg: ADMMforsrp} and assume that
	\begin{align*}
		\sum_{j=1}^{\infty} \| E_{\boldsymbol{z}}^{j+1}\|^2 < \infty \text{ and } \{ \boldsymbol{ z}^j\} \text{ is bounded}.
	\end{align*}
	Then, the sequence $ \{ \boldsymbol{x}^j,\boldsymbol{y}^j \}_{j=1}^{+\infty}$ is bounded and any limiting point $\{ \overline{\boldsymbol{x}},\overline{ \boldsymbol{y}} \}$ of $ \{ \boldsymbol{x}^j,\boldsymbol{y}^j \}_{j=1}^{+\infty}$ is a stationary point of problem in \eqref{eq: problemKKT}.
\end{theorem}
\begin{proof}
	See Appendix \ref{Prooftheorem1}.
\end{proof}
\begin{remark}
	Note that our problem in \eqref{eq: problem2} or \eqref{eq: problemKKT} is non-convex and non-smooth, which is in general almost impossible to get a global optimal solution. The above theorem means that any accumulation point that our algorithm generates is a critical point, which is suboptimal to our problem in \eqref{eq: problem2} or \eqref{eq: problemKKT}.
\end{remark}
%\begin{remark}
%    Due to the special form of our problem \eqref{eq: problemKKT}, the two
%    assumptions in Theorem \ref{thm:conv} can not be removed by simply following
%	the framework in \cite{attouch2010proximal,bolte2014proximal,hong2016convergence,chang2019blind} for the convergence analysis of an iterative
%	algorithm for non-convex optimization problems. 
%    On the other hand, let us add a new variable $\boldsymbol{\sigma} =
%    \boldsymbol{x}$ and rewrite our problem as 
%\begin{align*}
%\begin{split}
%\min_{\boldsymbol{y},\boldsymbol{\sigma},\boldsymbol{x}} \quad & f(\boldsymbol{x}) + g(\boldsymbol{y}) + \mathcal{I}_{ \mathcal{Y} }(\boldsymbol{y}) + \mathcal{I}_{\Sigma }(\boldsymbol{\sigma}),\\
%\mbox{s.t.}\quad
%& \hat{\boldsymbol{y}} = \hat{\boldsymbol{R}}\boldsymbol{x},    
%\end{split}
%% \label{eq: problemaddsigma}
%\end{align*}
%where $\hat{\boldsymbol{y}} = 
%\begin{pmatrix}
%\boldsymbol{ y}\\
%\boldsymbol{ \sigma}
%\end{pmatrix}, \hat{\boldsymbol{R}} =
%\begin{pmatrix}
%\boldsymbol{R}\\
%\boldsymbol{I}_P
%\end{pmatrix} $, $\Sigma = \{ \boldsymbol{\sigma}| \| \boldsymbol{\sigma}_k \|_0 \leq w_k; \boldsymbol{\sigma} \geq \boldsymbol{ 0} \}$.
%Then, one can establish a stronger
%	convergence result without the two
%    assumptions in Theorem \ref{thm:conv}  as in
%    \cite{attouch2010proximal,bolte2014proximal,hong2016convergence,chang2019blind}.
%    However, our numerical performance of the above reformulation is worse than
%    that of \eqref{eq: problemKKT}. %there are some obstacles in solving corresponding subproblems.
%\end{remark}

\section{Numerical Simulations}
\label{Sec: 6}
In this section, we perform numerical simulations to illustrate the performance gain of our algorithm by comparing several baselines. We first consider a baseline by adopting the convex relaxation approach in \cite{bi2016network} to solve our MOPC in \eqref{eq: MRproblem}. Specifically, we relax the $\ell_0$-norm constraints in \eqref{eq: MRproblem} to linear constraints using the relaxation technique in \cite{bi2016network}. First the nonlinear term $\max_l \frac{\boldsymbol{ R}[l]\boldsymbol{x}}{c_l}$ in the objective function is transformed to a linear objective function and some constraints by introducing an auxiliary variable $t$. After such linear relaxation and linear representation, \eqref{eq: MRproblem} is transformed as
\begin{align}
\begin{split}
\min_{\boldsymbol{x},t} \quad & -\sum_{k=1}^{K} U_k(\|\boldsymbol{x}_k\|_1) + \alpha t \\
\mbox{s.t.}\quad
&\boldsymbol{Rx} \leq t\boldsymbol{c}  \\
& \sum_{i=1}^{P_k} \frac{x_{k,i}}{\hat{c}_{k,i}} \leq w_k ,\quad \forall \, k \\
& \boldsymbol{x} \geq \boldsymbol{0} \\
& 0 \leq t \leq 1,
\end{split}
\label{eq: cvr3}
\end{align}
 where $\hat{c}_{k,i} = \min_l \{c_l | R_{l,i}^k =1 \}$ can be interpreted as the capacity of the bottleneck link along path $i$ of $k$-th flow \cite{bi2016network}. We solve the above problem via the conditional gradient method (also called the Frank-Wolfe method) \cite{jaggi2013revisiting}, which serves as one of our baselines.

\subsection{Dataset Descriptions} We perform our simulations based on the data obtained from practical real WAN. Specifically, the flow size vector $\boldsymbol{s}$, the link capacity vector $\boldsymbol{c}$, the routing matrix $\boldsymbol{R}$ and vector $\boldsymbol{w}$ (the number of paths each flow is allowed to use) are configured as follows:
\begin{itemize}
	\item $\boldsymbol{R}$: $\boldsymbol{R}\in \mathbb{R}^{L\times P}$ with $L=460$, and $P = 19751$ from a practical network.
	\item $P_k,\,(k=1,\cdots,K)$: Ranging from 4 to 100. The cumulative distribution function of $\{P_k\}_{k=1}^K$ is plotted in \figref{fig: pkdist}.
	\item $\boldsymbol{c}$: Ranging from $1.024 \times 10^9$ bits/sec to $2.048 \times 10^{11}$ bits/sec. The cumulative distribution function of link capacity is plotted in \figref{fig: cdist}.

	\item $\boldsymbol{s}$: We generate a total number of 561 aggregate flows and the size of each is constructed as follows:
	First we generate several flows (the number of flows is randomly chosen from $[10^3,10^4]$) according to one kind of flow size distribution randomly chosen from the  five typical
	flow size distributions \cite{montazeri2018homa} available, which are
	plotted in \figref{fig: flowdist}. Next those generated flows are aggregated to be one flow and the size of each aggregate flow is the total size of flows in that aggregate one. Our generated $\boldsymbol{s}$ is plotted in \figref{fig: sdist}.
	\item $\boldsymbol{w}$: We generate $w_k$ randomly from $\{1,2,3\}$.

\end{itemize}
\begin{figure}[]
	\centering
	\subfigure[The cumulative distribution function of the number of total available paths $P_k,\;(k=1,\cdots,K)$.]{
		\includegraphics[width=7cm]{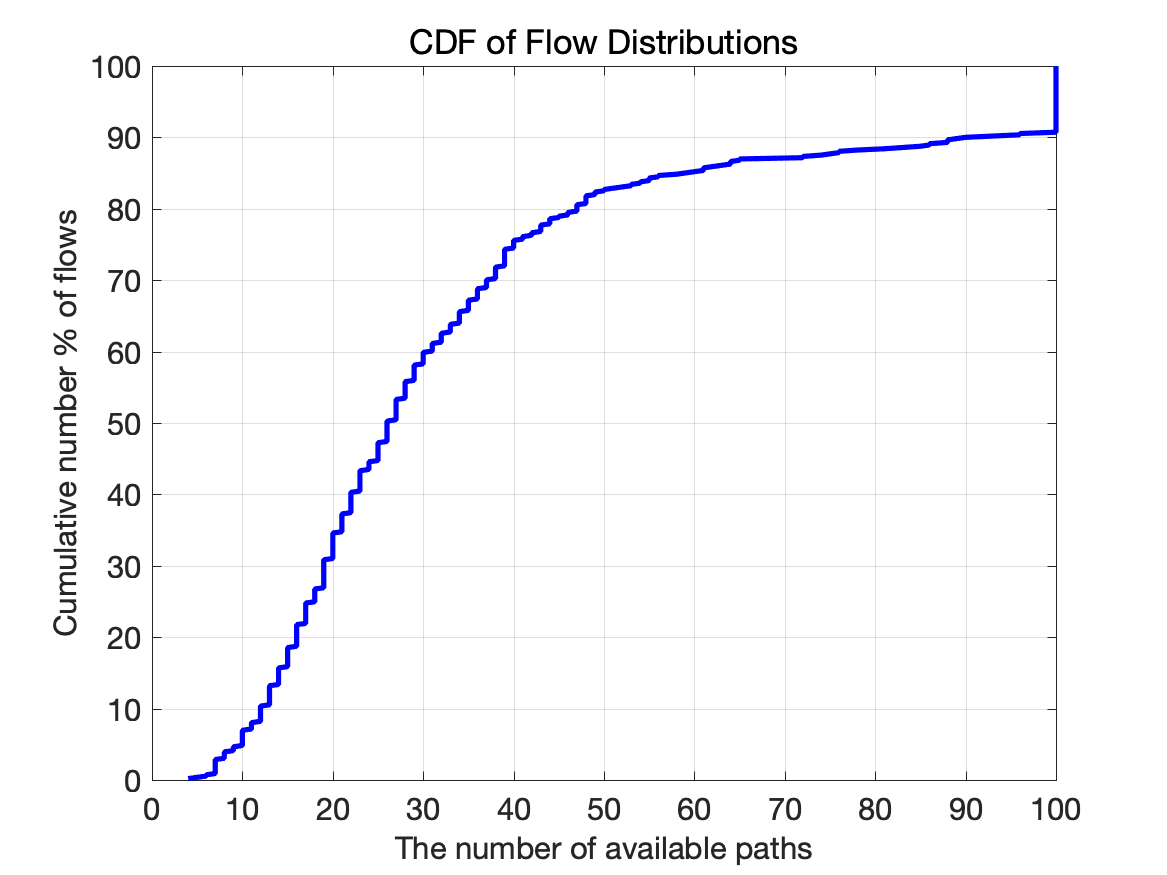}
		\label{fig: pkdist}   %标签，用作引用
	}
	\quad
	\subfigure[The cumulative distribution function of link capacity.]{
		\includegraphics[width=7cm]{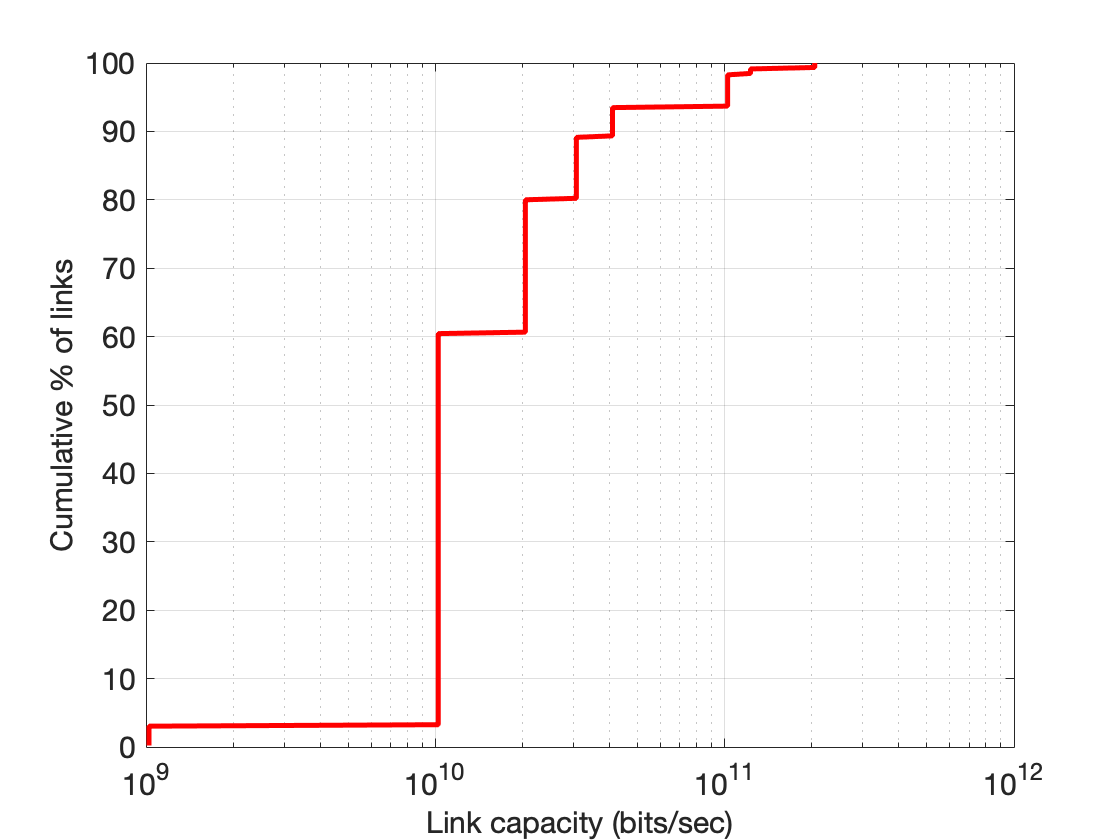}
		 \label{fig: cdist}   %标签，用作引用
	}
\caption{Illustrations of the total available paths and link capacity.}
\end{figure}

\begin{figure}[htbp]
	\centering
	\subfigure[Five types of flow distributions. F1: ``FacebookKeyValue\_Sampled"; F2: ``Google\_SearchRPC"; F3: ``Google\_AllRPC"; F4: ``Facebook\_HadoopDist\_All"; F5: ``DCTCP".]{
		\includegraphics[width=7cm]{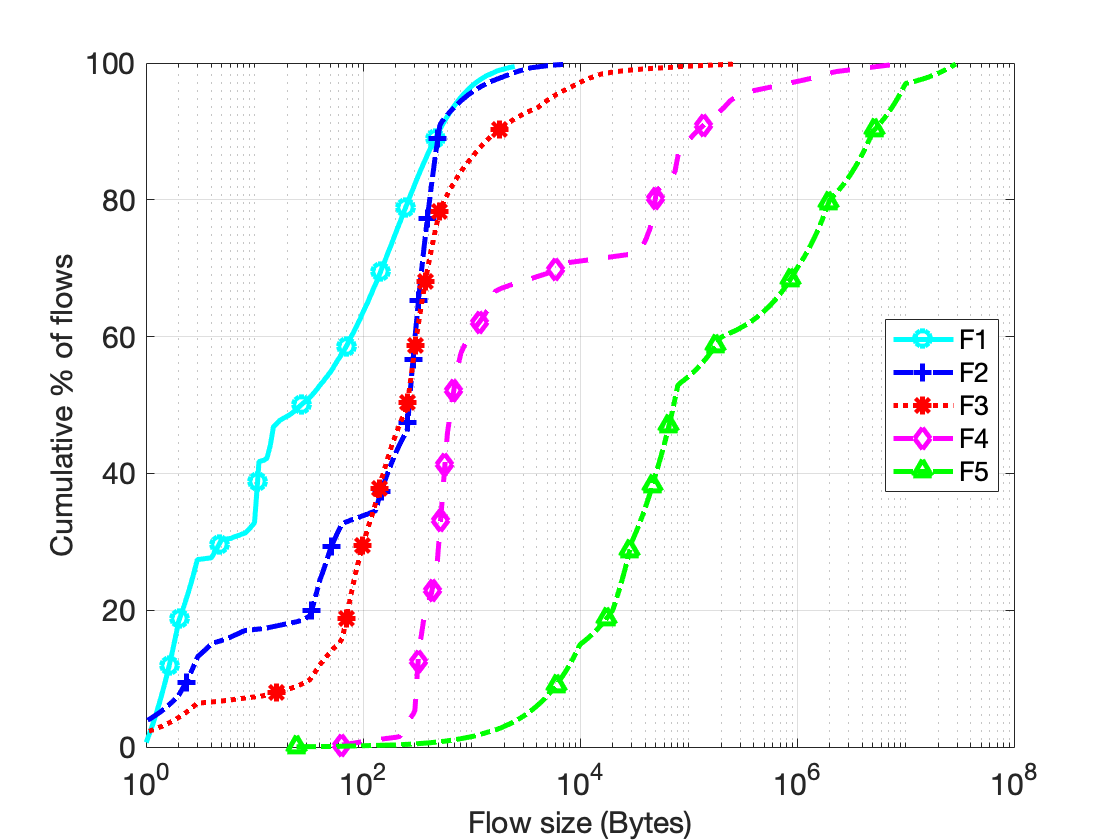}
\label{fig: flowdist}   %标签，用作引用
	}
	\quad
	
	\subfigure[The distribution of our generated flow sizes.]{
		\includegraphics[width=7cm]{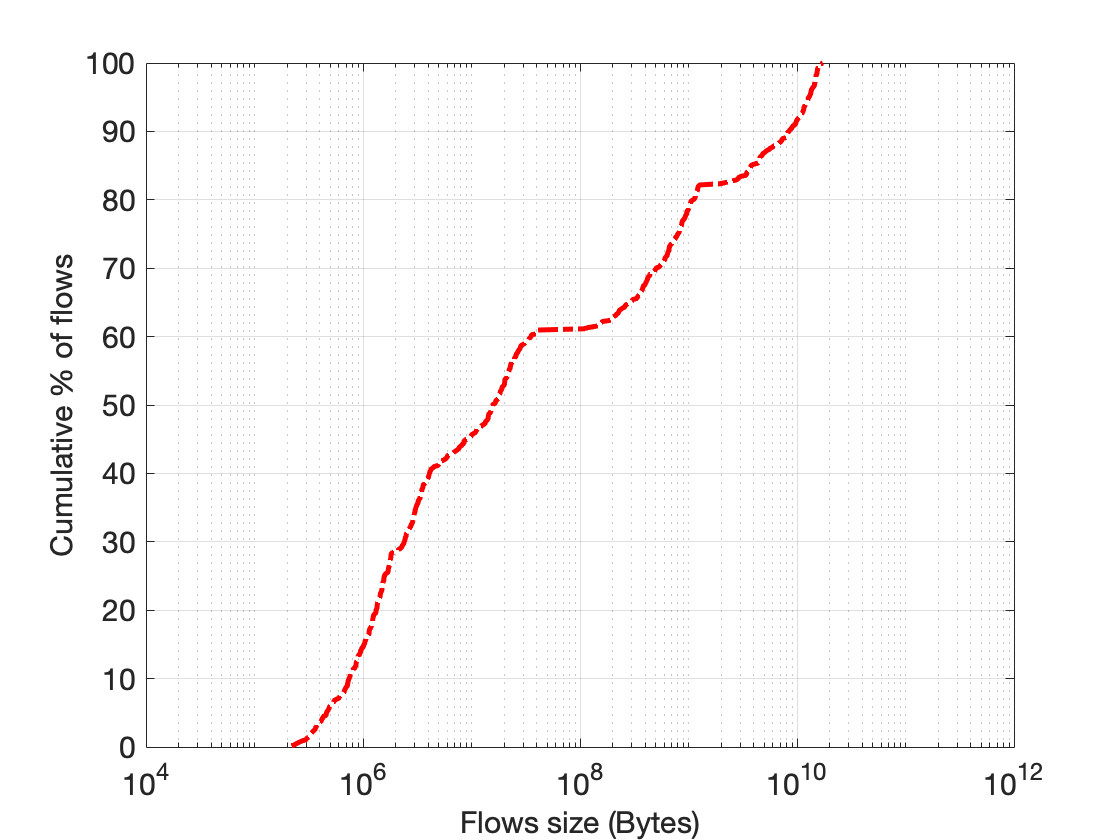}
		\label{fig: sdist}   %标签，用作引用
	}
\caption{Illustrations of flow distributions.}
\end{figure}

\subsection{Parameter Settings}
\label{subsec: ps}
 Our simulation is performed in MATLAB on a PC with an Intel Core i5 at 2.3GHz and 8GB of memory. For the parameters in our objective function, $\alpha$ and $\beta$ are set to be $500$ and $0.05$, respectively. 
The maximal iteration number is set to be $1500$. In ADMM for the convex scenario of MOPC \eqref{eq: MRproblem}, that is the case when $w_k = P_k, \forall k$, we adjust the parameter $\rho$ according to the Section 3.4.1 in \cite{boyd2011distributed}, and set the parameter $\mu =1.1 \times \rho \| \boldsymbol{R} \|_2^2$. In the $\boldsymbol{ z}$-update, we take an additional step length with $\gamma = 1.618$, which demonstrates better convergence performance. 
The stopping criteria are
\begin{itemize}
	\item primal residual ($p\_res$): 
	\begin{align*}
		 \| \boldsymbol{y}^j-\boldsymbol{Rx}^j\|_2 \leq \sqrt{L}\epsilon^{\text{abs}} + \epsilon^{\text{rel}} \max\{\| \boldsymbol{y}^j \|_2, \| \boldsymbol{R}\boldsymbol{x}^j \|_2 \},
	\end{align*}
	\item dual residual ($d\_res$): 
	\begin{align*}
		\| \rho \boldsymbol{R}^\top(\boldsymbol{y}^j -\boldsymbol{y}^{j-1} )   \|_2 \leq \sqrt{P}\epsilon^{\text{abs}} + \epsilon^{\text{rel}}  \| \boldsymbol{R}^\top\boldsymbol{z}^j \|_2,
	\end{align*}
	\item constraint violation ($vio$):
	\begin{align*}
		\| \max(\boldsymbol{Rx-c},\boldsymbol{0})\|_2/ \max(\sqrt{L},\| \boldsymbol{c}\|_2)\leq \epsilon^{tol},
	\end{align*}
\end{itemize}
where we set $\epsilon^{\text{abs}} = 10^{-4}, \epsilon^{\text{rel}} = 10^{-4}$, and
$\epsilon^{\text{tol}} = 10^{-10}$.

In ADMM for the non-convex scenario of MOPC \eqref{eq: MRproblem}, that is $\exists k, w_k < P_k$, due to non-convexity, we increase
$\rho$ during the iteration procedure and set $\mu = 1.1 \times \rho \|
\boldsymbol{R} \|_2^2$. The stopping criteria are
\begin{itemize}
	\item primal residual ($p\_res$): 
	\begin{align*}
	\| \boldsymbol{y}^j-\boldsymbol{Rx}^j\|_2 /\max(\sqrt{L}, \| \boldsymbol{y}^j\|_2)\leq \epsilon^{tol1},
	\end{align*}
	\item $\boldsymbol{y}$ difference ($y\_dif$): 
	\begin{align*}
\|  \boldsymbol{y}^j-\boldsymbol{y}^{j-1}\|_2/\max(\sqrt{L}, \| \boldsymbol{y}^{j-1}\|_2) \leq \epsilon^{tol1},
	\end{align*}
	\item constraint violation ($vio$):
	\begin{align*}
	\| \max(\boldsymbol{Rx-c},\boldsymbol{0})\|_2/ \max(\sqrt{L},\| \boldsymbol{c}\|_2)\leq \epsilon^{tol2},
	\end{align*}
\end{itemize}
where $\epsilon^{tol1}=10^{-4}$, and $\epsilon^{tol2}=10^{-10}$. We utilize 
\begin{itemize}
	\item the total completion time, namely delay,  
	\begin{align*}
		delay = \sum_{k=1}^{K} \frac{\boldsymbol{s}_k}{ \| \boldsymbol{x}_k  \|_1 },
	\end{align*}
	\item the proportional fairness,
		\begin{align*}
 fairness =  \beta \sum_{k=1}^{K} \log( \| \boldsymbol{x}_k  \|_1 ) ,
	\end{align*}
	\item the maximal (say the worst-case) link utilization ratio,
			\begin{align*}
load = \max_l \frac{\boldsymbol{R}[l] \boldsymbol{x} }{\boldsymbol{c}_l},
	\end{align*}
	\item the objective function value in MOPC problem,
	\begin{align*}
obj = delay  - fairness + \alpha \times load,
	\end{align*}
\end{itemize} as performance measures. Next, we present the experimental results in detail.
\subsection{Experimental Results}
We test several baselines listed in the following:
\begin{itemize}
	\item $cvxCG$: The conditional gradient method for our original problem \eqref{eq: MRproblem} without $\ell_0$-norm constraints.
	\item $rlxCG$: The conditional gradient method for the convex relaxation model in \eqref{eq: cvr3}.
	\item $cvxCG\_c$: Projecting the solution to $cvxCG$ to meet the $\ell_0$-norm constraints.
	\item $rlxCG\_c$: Projecting the solution to $rlxCG$ to meet the $\ell_0$-norm constraints.
	\item $cvxMOPC$: Our ADMM for our original problem \eqref{eq: MRproblem} without $\ell_0$-norm constraints (because $w_k = P_k, \; \forall k$).
	\item $ncvMOPC$: Our ADMM for the non-convex problem \eqref{eq: MRproblem}.
\end{itemize}

As shown in Table \ref{tab: res}, we verify the advantages of our method. First, the objective function value $obj$ and
other three performance measures $delay$, $fairness$, and $load$ in the
convex relaxation model $rlxCG$ are nearly the same as those in the original
convex model $cvxCG$. Therefore, the convex relaxation technique in \cite{bi2016network} does not work in our problem in \eqref{eq: MRproblem}.
Second, comparing the results of $cvxCG\_c$ with those of $rlxCG\_c$ shows that
there are indeed some improvements of this convex relaxation technique after
projecting its solution to meet the path cardinality constraints. In addition, the solution of our non-convex ADMM performs the best compared with other two projected solutions. \figref{fig: admmconverge} illustrates the convergence procedure of our non-convex ADMM. It can be observed that our method $ncvMOPC$ gets significant performance gain in $obj$, $delay$ as well as $fairness$ with tiny cost in $load$.

\begin{table}[tbhp] 
	\renewcommand\arraystretch{1.2}
	\caption{Results of MOPC using ADMM VS convex relaxation.}
	\centering
	\label{tab: res}
	
\begin{tabular}{|p{1.5cm}|p{0.7cm}|p{1.0cm}|p{1.2cm}|p{0.7cm}|} \hline  
	Scheme & Obj & Delay & Fairness & Load \\ \hline 
	cvxMOPC  & 121 & 329 & 561 & 0.70    \\\hline 
	cvxCG  & 119 & 322 & 564 & 0.72    \\\hline 
	rlxCG  & 119 & 321 & 564 & 0.72    \\\hline 
	cvxCG\_c  & 391 & 584 & 554 & 0.72    \\\hline 
	rlxCG\_c  & 369 & 562 & 555 & 0.72    \\\hline 
	ncvMOPC  & \textbf{169} & \textbf{361} & \textbf{562} & 0.74    \\\hline 
\end{tabular} 

\end{table}

	\begin{figure}[]%%图
	\centering  %插入的图片居中表示
	\includegraphics[scale=0.35]{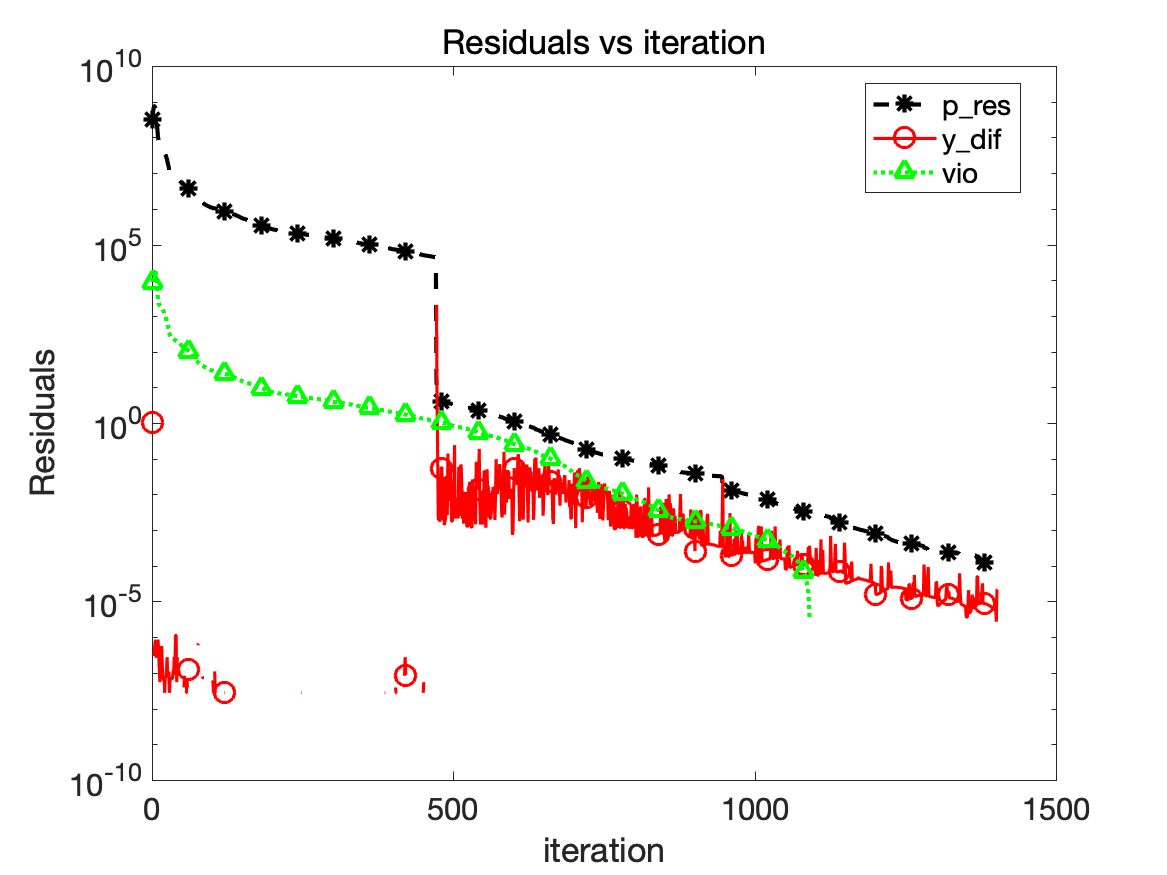} 
	\caption{The convergence procedure of our non-convex ADMM for MOPC.}  %图片的名称
	\label{fig: admmconverge}   %标签，用作引用
\end{figure}

\section{Summary}
\label{Sec: 7}
This paper investigates a multi-objective bandwidth optimization and path selection problem for
any given path cardinality constraint, which is in general a highly non-convex problem. An ADMM-based algorithm is proposed. The algorithm is
simple and easy to be implemented, whose subproblems include finding the maximal root of a single-cubic equation and a one-dimensional optimization problem. We validate the effectiveness of our algorithm both theoretically and experimentally. 

%There are a few interesting problems to be addressed in the future, and some of them are list as follows.
%\begin{enumerate}
%	\item How to extend our formulation and method to the dynamic (or online) scenario where flows randomly arrive in the network but still with a fixed network topology?
%	\item Since the routing matrix is a sparse matrix, is it possible to accelerate our algorithm via utilizing this kind of characteristics?
%%	\item Whether the non-convex block coordinate descent method works in our subproblem with $\ell_0$-norm constraints?
%\end{enumerate}

\appendix

\section{Proofs}
\subsection{Proof of Lemma \ref{lemma: xupdateinfty}}
\label{Prooflemma2}
\begin{proof}
	We reformulate \eqref{eq: NUMx-subproblemblock2} as
	\begin{align}
	\boldsymbol{x}^{j+1}_k= \mathop{\arg\min}_{ \boldsymbol{x}_k \geq \boldsymbol{0}} \quad  -U_k( \sum_{i=1}^{P_k} x_{k,i} ) +  \dfrac{\mu}{2}\| \boldsymbol{x}_k - \boldsymbol{\nu}_k^j  \|^2.
	\label{x-subproblemblock1}
	\end{align}
	The Lagrangian function of the optimization problem in \eqref{x-subproblemblock1} is 
	\begin{align*}
	L(\boldsymbol{x}_k;\boldsymbol{\lambda}_k) = -U_k( \sum_{i=1}^{P_k} x_{k,i} ) +  \dfrac{\mu}{2}\| \boldsymbol{x}_k - \boldsymbol{\nu}_k^j  \|^2 - \boldsymbol{\lambda}_k^\top \boldsymbol{x}_k,
	\end{align*}
	where $\boldsymbol{\lambda}_k = (\lambda_{k,1},\lambda_{k,2},\cdots,\lambda_{k,P_k})^\top \geq \boldsymbol{0}$ is a Lagrangian multiplier vector.
	From the KKT conditions of \eqref{x-subproblemblock1}, we can get 
	\begin{subequations}
		\begin{align}
		& x_{k,i}^{j+1} =  0 \quad \text{or}\quad \nu_{k,i}^j + \zeta, \quad \forall \, i, \label{eq: xupdateinfty1} \\
		& \boldsymbol{x}_{k}^{j+1} \geq \boldsymbol{ 0} ,  \label{eq: xupdateinftypos} \\
		& \mu \zeta = U'_k\left( \sum_{i=1}^{P_k} x_{k,i}^{j+1} \right), \label{eq: xupdateinfty2} 
		\end{align}
	\end{subequations}
where $\zeta$ is some positive number.
	Next, we will show that \eqref{eq: xupdateinfty1} and \eqref{eq: xupdateinftypos} are actually 
	\begin{align}
	x_{k,i}^{j+1} =  \max(0, \nu_{k,i}^j + \zeta), \quad \forall \, i. \label{eq: xupdateinftyst1} 
	\end{align}
	Suppose there exists some index $i\, \left(i\in \{1,2,\cdots, P_k\} \right) $ such that
	\begin{align*}
	\nu_{k,i}^j + \zeta > 0 \quad  \text{and} \quad x_{k,i}^{j+1} = 0.
	\end{align*}
	Then the gradient with respect to $\tau$ of the following function
	\begin{align}
	-U_k(t + \tau ) + \frac{\mu}{2}( \tau - \nu_{k,i}^j)^2,
	\label{eq: xupdateinftymono} 
	\end{align}
	where $t$ is the sum of all elements of $ \boldsymbol{x}_k^{j+1}$ except $x_{k,i}^{j+1}$, in $\tau = 0 $ is 
	\begin{align*}
	- U_k'( t + 0 ) + \mu (0 - \nu_{k,i}^j )
	= -\mu \zeta - \mu \nu_{k,i}^j < 0.
	\end{align*}
	It means that if we increase $ x_{k,i}^{j+1}$ from 0, our objective function value would decrease, 
	which contradicts the fact that $x_{k,i}^{j+1} = 0$ is optimal. 
	Then we have $x_{k,i}^{j+1} =  \nu_{k,i}^j + \zeta $, and readily get 
	\begin{align*}
	x_{k,i}^{j+1} =  \max(0, \nu_{k,i}^j + \zeta), \quad \forall \, i.
	\end{align*}
	Combining \eqref{eq: xupdateinfty2} with \eqref{eq: xupdateinftyst1} yields
	\begin{align}
	\mu \zeta = U'_k \left( \sum_{i=1}^{P_k} \max(0, \nu_{k,i}^j + \zeta) \right) . \label{eq: xupdateinfty3}
	\end{align}
	In the following, we will show how to get the value of $\zeta$. Without loss of generality, suppose the elements of $\boldsymbol{ \nu}_k^j$ are in descending order.
	Since the right hand side of \eqref{eq: xupdateinfty3} decreases with $\zeta$
	increasing and the left hand side increases with $\zeta$ increasing, we can try every
	positive $\zeta = -\nu_{k,i}^j \, (i=2,\cdots, P_k) $ to determine the corresponding
	interval in which $\zeta$ satisfies \eqref{eq: xupdateinfty3}. Specifically, if there exists some smallest index $i'$ $(i' \in {2,\cdots, P_k}) $ such that
		\begin{align}
	\mu \zeta \geq U'_k \left( \sum_{i=1}^{P_k} \max(0, \nu_{k,i}^j + \zeta) \right), \label{eq: xupdateinftynn}
	\end{align}
	where $\zeta = -\nu_{k,i'}^j $,
	then the corresponding interval is $(-\nu_{k,i'-1}^j, -\nu_{k,i'}^j ]$.
The corresponding interval is $(-\nu^j_{k,P_k},+\infty)$ otherwise.
	Then, we can obtain $\zeta$ through finding the maximal root of a single-variable cubic equation
			\begin{align*}
	\mu \zeta \left(\sum_{i=1}^{i'-1} (\nu^j_{k,i}+\zeta) \right)^2 = \beta \left(\sum_{i=1}^{i'-1} (\nu^j_{k,i}+\zeta) \right) + s_k,
	\end{align*}
	or
				\begin{align*}
	\mu \zeta \left(\sum_{i=1}^{P_k} (\nu^j_{k,i}+\zeta) \right)^2 = \beta \left(\sum_{i=1}^{P_k} (\nu^j_{k,i}+\zeta) \right) + s_k,
	\end{align*}
	otherwise.

\end{proof}
\subsection{Proof of Lemma \ref{lemma: xupdate1}}
\label{Prooflemma1}
\begin{proof}
	If $w_k=1$, $\boldsymbol{x}_k^{j+1}$ can have one and only one non-zero
    element. Next, we would first show the index of this non-zero element and
    then determine its corresponding value. Suppose there exist $i, i'\, (i,i' \in
    \{1,\cdots,P_k  \})$ such that $\nu_{k,i'}^j< \nu_{k,i}^j$, then for $\forall a >0$, we have
	\begin{align*}
	& -U_k(a) + \frac{\mu}{2}(a-\nu_{k,i'}^j)^2 + \frac{\mu}{2}(\nu_{k,i}^j)^2	\\
	&> -U_k(a) + \frac{\mu}{2}(a-\nu^j_{k,i})^2 + \frac{\mu}{2}(\nu^j_{k,i'})^2,
	\end{align*}
	meaning that $( \overbrace{ 0,\cdot\cdot\cdot,0
    }^{i'-1},a,0,\cdot\cdot\cdot,0)^\top$ is less optimal than $( \overbrace{
    0,\cdot\cdot\cdot,0 }^{i-1},a,0,\cdot\cdot\cdot,0)^\top$. Hence, the index of the non-zero element is $i$ $(i \in \{1,\cdots,P_k  \})$ that makes $\nu^j_{k,i}$ be maximal. Now problem \eqref{x-subproblemblock} is converted into a single-variable optimization problem:
	\begin{align}
	x^{j+1}_{k,i}  &= \mathop{\arg\min}_{ a } -U_k(a) + \dfrac{\mu}{2} ( a -  \nu^j_{k,i} )^2.
	\label{x-update_single}
	\end{align}
	To make the gradient of problem \eqref{x-update_single} be zero, we obtain a single-variable cubic equation
	\begin{align*}
	\mu a^3 - \mu \nu^j_{k,i}   a^2 - \beta a -s_k/K  =0, \quad   \forall k.
	\end{align*}
	$x_{k,i}^{j+1}$ is the maximal real positive root (there is at least one real positive root because $- s_k/K < 0$).
\end{proof}

\subsection{Proof of Lemma \ref{lemma: xupdate2}}
\label{Prooflemma3}
\begin{proof}
	First, when $1<w_k<P_k$, we prove the solution to \eqref{x-subproblemblock} satisfies that if $i > w_k $, $x_{k,i}^{j+1} = 0$ by contradiction. 
	
	Suppose one solution $\tilde{\boldsymbol{ x}}_k^{j+1}$ satisfies that there exists $i > w_k$ such that $\tilde{x}_{k,i}^{j+1} = d > 0$, then there exists $i' \leq w_k$ such that $\tilde{x}_{k,i'}^{j+1} = 0$, i.e. 
		$$ \tilde{\boldsymbol{ x}}_k^{j+1} = (\tilde{x}_{k,1}^{j+1},\cdots, {\tilde{x}_{k,i'}^{j+1}=0}, \cdots, {\tilde{x}_{k,i}^{j+1}=d},\cdots, \tilde{x}_{k,P_k}^{j+1} )^\top.$$
		We can construct another feasible point by switching the values of $\tilde{x}_{k,i'}^{j+1}$ and $\tilde{x}_{k,i}^{j+1}$. Then we get 
		$$ \hat{\boldsymbol{x}}_k^{j+1} = (\tilde{x}_{k,1}^{j+1},\cdots, {\hat{x}_{k,i'}^{j+1}=d}, \cdots, {\hat{x}_{k,i}^{j+1}=0},\cdots, \tilde{x}_{k,P_k}^{j+1} )^\top.$$
		Because of $\nu^j_{k,i} \leq \nu^j_{k,i'} $, we have
		 \begin{align*}
			&-U_k(t + d ) + \frac{\mu}{2}(d- \nu^j_{k,i'})^2 +  \frac{\mu}{2}( \nu^j_{k,i})^2  \nonumber\\ &\leq -U_k(t + d ) + \frac{\mu}{2}( \nu^j_{k,i'})^2 +  \frac{\mu}{2}( d - \nu^j_{k,i})^2,
			\end{align*}
		where $t$ is the sum of all elements of $\boldsymbol{ x}_{k}^{j+1}$ except $x_{k,i}^{j+1}$ and $x_{k,i'}^{j+1}$. 
        Hence, $\hat{\boldsymbol{x}}_k^{j+1}$ is better than
        $\tilde{\boldsymbol{ x}}_k^{j+1}$. By repeating this procedure, we get our conclusion $x_{k,i}^{j+1} = 0, \; i > w_k$.
		
	Now problem \eqref{x-subproblemblock} can be transformed to 
	\begin{align}
	& (x^{j+1}_{k,1}, \cdots,x^{j+1}_{k,w_k} )^\top = \nonumber \\
	&  \mathop{\arg\min}_{ x_{k,i} \geq 0, i \leq w_k} \quad  -U_k( \sum_{i=1}^{w_k} x_{k,i} ) +  \dfrac{\mu}{2} \sum_{i=1}^{w_k} ( x_{k,i} - \nu^j_{k,i} )^2.
	\label{eq: lemmax22}
	\end{align}
	From similar analysis in the proof of lemma \ref{lemma: xupdateinfty}, the solution to \eqref{eq: lemmax22} satisfies \eqref{eq: lemmaxs1}
    and \eqref{eq: lemmaxs2}. This completes the proof. 
	
\end{proof}

\subsection{Proof of Lemma \ref{lemma: y-subproblem}}
\label{Prooflemma4}
\begin{proof}
	Consider problem \eqref{eq: y-updatesub}. For a fixed $t$, the optimal $\boldsymbol{y}$ is shown in \eqref{eq: ysubsolution}.
	If $t^*$ is the optimal solution to
	\begin{align*}
	\min_t  \Phi(t) = \phi(t,\boldsymbol{y}(t))  \quad s.t. \quad t \geq 0,
	\end{align*}
	then $\boldsymbol{y}(t^*)$ is the optimal solution to problem  \eqref{eq: scaleadmmpro2} (since for any other feasible $t'$ and $\boldsymbol{y}'$, we have 
	\begin{align*}
		\phi(t',\boldsymbol{y}') \geq \phi(t',\boldsymbol{y}(t')) = \Phi(t')\geq \Phi(t^*) = \phi(t^*,\boldsymbol{y}(t^*)).
	\end{align*}
	 By noticing that the function 
	\begin{align*}
	\Phi(t)= \phi(t,\boldsymbol{y}(t)) = \inf_{\boldsymbol{y}\in \mathcal{C}} \, \phi(t,\boldsymbol{y}),
	\end{align*}
	is convex in $t$ (because $\phi$ is convex in $(t,\boldsymbol{y})$, and
    $\mathcal{C}$ is a convex nonempty set), and when $t>1$, $\Phi(t)$ will
    increase. This concludes the proof.
\end{proof}

\subsection{Proof of Proposition \ref{propo: alg2albound}}
\label{Proofpropo1}
\begin{proof}
	We recast the augmented Lagrangian function \eqref{eq: algconv} into another form:
	\begin{align*}
	L_\rho(\boldsymbol{x}^j,\boldsymbol{y}^j;\boldsymbol{z}^j)  = f(\boldsymbol{x}^j) + g(\boldsymbol{y}^j) + \mathcal{I}_{ \mathcal{X} }(\boldsymbol{x}^j) 	+ \mathcal{I}_{ \mathcal{Y}}(\boldsymbol{y}^j)  \\
	+ \frac{\rho}{2} \| \boldsymbol{y}^j-\boldsymbol{R}\boldsymbol{x}^j + \frac{\boldsymbol{z}^j}{\rho} \|^2 - \frac{\| \boldsymbol{ z}^j \|^2}{2\rho}.
	\end{align*}
	Since when $\boldsymbol{x}^j_{k,i} \rightarrow \infty\,(k\in \{1,2,\cdots,K\}, \forall i \in \{1,2,\cdots,P_k \})$, we have $ \| \boldsymbol{Rx}^j\| \rightarrow +\infty$, which combining with the boundedness of $\{\boldsymbol{y}^j \} $ and $\{\boldsymbol{z}^j \}$ implies that $f(\boldsymbol{x}^j) + \frac{\rho}{2} \| \boldsymbol{y}^j-\boldsymbol{R}\boldsymbol{x}^j + \frac{\boldsymbol{z}^j}{\rho} \|^2 \rightarrow +\infty$. Besides, other parts of  $	L_\rho(\boldsymbol{x}^j,\boldsymbol{y}^j;\boldsymbol{z}^j) $ are all bounded below. We get the lower boundedness of $\{L_\rho(\boldsymbol{x}^j,\boldsymbol{y}^j;\boldsymbol{z}^j)\}$.
\end{proof}

\subsection{Proof of Lemma \ref{lemma: weakdc}}
\label{Prooflemma5}
\begin{proof}
	For simplicity, we denote $h(\boldsymbol{x}) = \frac{\rho}{2} \| \boldsymbol{y}^j-\boldsymbol{Rx}+ \boldsymbol{z}^j/\rho  \|^2 $,
	then the $\boldsymbol{x}$-update \eqref{x-update1} is actually
	\begin{align*}
	\boldsymbol{x}^{j+1} = \mathop{\arg\min}_{ \boldsymbol{x} \in \mathcal{X} } f(\boldsymbol{x}) + \langle \nabla h(\boldsymbol{x}^j), \boldsymbol{x} - \boldsymbol{x}^j \rangle  + \frac{\mu}{2}\| \boldsymbol{x} - \boldsymbol{x}^j \|^2,
	\end{align*}
	which yields that 	
	\begin{align}
	& f(\boldsymbol{x}^j) \geq  \nonumber \\
	& f(\boldsymbol{x}^{j+1}) + \langle \nabla h(\boldsymbol{x}^j), \boldsymbol{x}^{j+1}-\boldsymbol{x}^j \rangle + \frac{\mu}{2}\| \boldsymbol{x}^{j+1} - \boldsymbol{x}^j \|^2 .
	\label{eq: convergexsub1}
	\end{align}
	Since $ \| \nabla h(\boldsymbol{x}) - \nabla h(\boldsymbol{x}')  \| \leq \rho \| \boldsymbol{R} \|_2^2 \| \boldsymbol{x} - \boldsymbol{x}'  \| , \forall \boldsymbol{x}, \boldsymbol{x}' $, from Lemma 1 in \cite{bolte2014proximal}, we have 
	\begin{align}
	& h(\boldsymbol{x}^{j+1}) \leq h(\boldsymbol{x}^j) \nonumber \\
	& + \langle \nabla h(\boldsymbol{x}^j), \boldsymbol{x}^{j+1}-\boldsymbol{x}^j  \rangle + \frac{\rho \| \boldsymbol{R} \|_2^2}{2} \| \boldsymbol{x}^{j+1}  - \boldsymbol{x}^j \|^2.
	\label{eq: convergexsub2}
	\end{align}
	Summing up \eqref{eq: convergexsub1} and \eqref{eq: convergexsub2} gives
	\begin{align}
	&L_\rho(\boldsymbol{x}^j, \boldsymbol{y}^j; \boldsymbol{z}^j) - L_\rho(\boldsymbol{x}^{j+1},\boldsymbol{y}^j; \boldsymbol{z}^j)  \nonumber \\
	&= f(\boldsymbol{x}^j) + h(\boldsymbol{x}^j) -f(\boldsymbol{x}^{j+1}) - h(\boldsymbol{x}^{j+1})  \nonumber \\
	&\geq \frac{\mu - \rho \| \boldsymbol{R} \|_2^2}{2} \| E_{\boldsymbol{x}}^{j+1}  \|^2. 	\label{eq: weakdcx}
	\end{align}
	The $\boldsymbol{y}$-update \eqref{eq: scaleadmmpro2} is
	\begin{align*}
	\boldsymbol{y}^{j+1} = \mathop{\arg\min}_{\boldsymbol{0} \leq \boldsymbol{y}\leq \boldsymbol{c}}  g(\boldsymbol{y})+\dfrac{\rho}{2} \| \boldsymbol{y} - \boldsymbol{Rx}^{j+1} +  \boldsymbol{z}^j /\rho  \|^2,
	\end{align*}
	whose optimality condition deduces that $ \exists\, \boldsymbol{g}^{j+1} \in  \partial g(	\boldsymbol{y}^{j+1})$ such that
	\begin{align}
	\left( \boldsymbol{g}^{j+1}+ \rho ( \boldsymbol{y}^{j+1} -\boldsymbol{Rx}^{j+1} + \boldsymbol{z}^j /\rho ) \right) ^\top ( \boldsymbol{y}' -\boldsymbol{y}^{j+1} ) \geq 0,
	\label{eq: convergeopt}
	\end{align}
	where $\boldsymbol{0}  \leq \boldsymbol{y}' \leq \boldsymbol{c}$.
	Then we have
	\begin{align}
	& L_\rho(\boldsymbol{x}^{j+1}, \boldsymbol{y}^j; \boldsymbol{z}^j) - L_\rho(\boldsymbol{x}^{j+1},\boldsymbol{y}^{j+1}; \boldsymbol{z}^j)\nonumber \\
	& = g(	\boldsymbol{y}^j) - g(	\boldsymbol{y}^{j+1} ) \nonumber \\
	& + \frac{\rho}{2} \|   \boldsymbol{y}^j - \boldsymbol{Rx}^{j+1} +  \boldsymbol{z}^j /\rho    \|^2 - \frac{\rho}{2} \|   \boldsymbol{y}^{j+1} - \boldsymbol{Rx}^{j+1} +  \boldsymbol{z}^j /\rho    \|^2 \nonumber \\
	& = g(	\boldsymbol{y}^j) - g(	\boldsymbol{y}^{j+1} )  \nonumber\\
	& + \frac{\rho}{2} \| \boldsymbol{y}^{j} - \boldsymbol{y}^{j+1}  \|^2 + \rho \langle  \boldsymbol{y}^{j+1} - \boldsymbol{Rx}^{j+1} + \boldsymbol{z}^j/\rho,  \boldsymbol{y}^{j} - \boldsymbol{y}^{j+1}  \rangle \nonumber\\
	& \geq  (\boldsymbol{g}^{j+1})^\top ( 	\boldsymbol{y}^j - 	\boldsymbol{y}^{j+1}) \nonumber \\
	&+ \frac{\rho}{2} \| \boldsymbol{y}^{j} - \boldsymbol{y}^{j+1}  \|^2 + \rho \langle  \boldsymbol{y}^{j+1} - \boldsymbol{Rx}^{j+1} + \boldsymbol{z}^j/\rho,  \boldsymbol{y}^{j} - \boldsymbol{y}^{j+1}  \rangle  \nonumber\\
	& \overset{\eqref{eq: convergeopt}}{\geq} \frac{\rho}{2} \| E_{\boldsymbol{y}}^{j+1}  \|^2, \label{eq: weakdcy}
	\end{align}
	where the second equality is derived from 
	\begin{align*}
		\|\boldsymbol{v}\|^2 - \|\boldsymbol{w}\|^2 = \|\boldsymbol{v} -\boldsymbol{w}\|^2 + 2 \langle \boldsymbol{w},\boldsymbol{v} -\boldsymbol{w}  \rangle, \; \forall \boldsymbol{ v},\boldsymbol{w} \in \mathbb{R}^L.
	\end{align*}
	For the $\boldsymbol{z}$-update, we have
	\begin{align}
	& L_\rho(\boldsymbol{x}^{j+1}, \boldsymbol{y}^{j+1}; \boldsymbol{z}^j) - L_\rho(\boldsymbol{x}^{j+1},\boldsymbol{y}^{j+1}; \boldsymbol{z}^{j+1}) \nonumber \\
	& = (\boldsymbol{z}^j - \boldsymbol{z}^{j+1})^\top (\boldsymbol{y}^{j+1} - \boldsymbol{Rx}^{j+1}) \nonumber \\
	& =  - \frac{1}{\rho} \| E_{\boldsymbol{z}}^{j+1} \|^2.\label{eq: weakdcz}
	\end{align}
	By summing up \eqref{eq: weakdcx}, \eqref{eq: weakdcy}, and \eqref{eq: weakdcz}, we complete the proof.
\end{proof}

\subsection{Proof of Theorem \ref{thm:conv}}
\label{Prooftheorem1}
\begin{proof}
	From $\sum_{j=1}^{\infty} \| E_{\boldsymbol{z}}^{j+1} \|^2 < \infty$, we can get $\lim\limits_{j \rightarrow +\infty} E_{\boldsymbol{z}}^{j+1} = \boldsymbol{0}$.
	Next, we first show that
	\begin{align}
	\lim\limits_{j \rightarrow +\infty} E_{\boldsymbol{x}}^{j+1} = \boldsymbol{ 0}\; \text{and}\;  		\lim\limits_{j \rightarrow +\infty} E_{\boldsymbol{y}}^{j+1} = \boldsymbol{ 0}.
	\label{eq: xyinfty}
	\end{align}
	Since $\{ L_\rho(\boldsymbol{x}^{j}, \boldsymbol{y}^{j}; \boldsymbol{z}^j) \}$ is bounded below and $\sum_{j=1}^{\infty} \| E_{\boldsymbol{z}}^{j+1} \|^2 < \infty $, it follows from \eqref{eq: weakdc} that 
	\begin{align}
	\sum_{j=1}^{\infty} \| E_{\boldsymbol{x}}^{j+1}\|^2 < \infty,\, \sum_{j=1}^{\infty} \| E_{\boldsymbol{y}}^{j+1} \|^2 < \infty,
	\end{align}
	which implies \eqref{eq: xyinfty}. The boundedness of  $ \{
    \boldsymbol{x}^j\}_{j=1}^{+\infty}$ can be obtained from the boundedness of
    $ \{ \boldsymbol{y}^j \}_{j=1}^{+\infty}$, $ \{ \boldsymbol{z}^j \}_{j=1}^{+\infty}$ as well as  the $\boldsymbol{
    z}$-update in \eqref{eq: scaleadmmpro3}. In addition, because of the
    $\boldsymbol{ z}$-update in \eqref{eq: scaleadmmpro3}, we can have 
	\begin{align}
	\lim\limits_{j\rightarrow\infty}\boldsymbol{y}^j - \boldsymbol{ R}\boldsymbol{x}^j = 0. \label{eq: yeqrx}
	\end{align}
	
	For any limiting point $ \{ \overline{\boldsymbol{x}},\overline{ \boldsymbol{y}} \}$ of the sequence, since $\{\boldsymbol{x}^j,\boldsymbol{y}^j,\boldsymbol{z}^j  \}$ is bounded, there is a sequence $\{\boldsymbol{x}^{j_i},\boldsymbol{y}^{j_i},\boldsymbol{z}^{j_i} \}$ whose limiting point is $ \{ \overline{\boldsymbol{x}},\overline{ \boldsymbol{y}} ,\overline{ \boldsymbol{z}}\}$. 
	Clearly, from \eqref{eq: yeqrx}, we can get 
	\begin{align*}
	\overline{\boldsymbol{y}} = \boldsymbol{ R} \overline{\boldsymbol{x}},
	\end{align*}
	which is in fact \eqref{eq: KKTcon} in KKT conditions. Taking limit on both sides of \eqref{eq: optx}, \eqref{eq: opty2} and applying the Remark 1 in \cite{bolte2014proximal}, we obtain \eqref{eq: KKTx}, \eqref{eq: KKTy}  in the KKT conditions.
\end{proof}
%
%
%
%% you can choose not to have a title for an appendix
%% if you want by leaving the argument blank
%\section{Convergence Analysis}

%\section*{Acknowledgment}
%
%
%%We would like to thank the two reviewers and the associate editor for their valuable comments, which did great help to enhance the paper. 
%We would  like to thank Yingjie Bi for discussing related works and 
%Jinshan Zeng for discussing the convergence
% of ADMM for non-convex optimization problems.

\ifCLASSOPTIONcaptionsoff
  \newpage
\fi


\bibliographystyle{IEEEtran}
%\bibliography{IEEEabrv,mybibfile}
\bibliographystyle{mybstfile}
\bibliography{IEEEabrv,mybibfile}

\end{document}